\documentclass[paper=a4, fontsize=11pt]{scrartcl}
\usepackage[T1]{fontenc}
\usepackage{fourier}
\usepackage[english]{babel}								 
\usepackage{amssymb}
\usepackage[protrusion=true,expansion=true]{microtype}	
\usepackage{amsmath,amsfonts,amsthm} 
\usepackage[pdftex]{graphicx}	
\usepackage{url}
\usepackage[title,titletoc,page]{appendix} 
\usepackage{listings} 
\usepackage{color}
\usepackage{cite}
\usepackage{caption}
\usepackage{subcaption}
\usepackage[ruled,vlined]{algorithm2e}
\numberwithin{equation}{section}		
\numberwithin{figure}{section}			
\numberwithin{table}{section}				

\newtheorem{defi}{Definition}[section]
\newtheorem{lem}{Lemma}[section]
\newtheorem{thm}{Theorem}[section]

\newtheorem{remark}{Remark}[section]

\newtheorem{prop}{Proposition}[section]

\newcommand{\btwonorm}[1]{\Big|\Big|{#1}\Big|\Big|_2}
\newcommand{\vect}[1]{\boldsymbol{\mathbf{#1}}}

\title{A measurement decoupling based fast algorithm for super-resolving point sources with multi-cluster structure}
\author{
Ping Liu\thanks{\footnotesize 
  Department of Mathematics,
ETH Z\"{u}rich, Swizerland (ping.liu@sam.math.ethz.ch).}
 \; and Hai Zhang\thanks{\footnotesize 
 Department of Mathematics, 
  HKUST,  Clear Water Bay, Kowloon, Hong Kong, S.A.R, China (haizhang@ust.hk). Hai Zhang was partially supported by Hong Kong RGC grant GRF 16305419 and 16304621.}}

\begin{document}

\maketitle
\begin{center}
	\textbf{Abstract}
\end{center}
We consider the problem of resolving closely spaced point sources in one dimension from their Fourier data in a bounded domain. Classical subspace methods (e.g., MUSIC algorithm, Matrix Pencil method, etc.) show great superiority in resolving closely spaced sources, but their computational cost is usually heavy. This is especially the case for point sources with multi-cluster structure which requires processing large sized data matrix resulted from highly sampled measurement. To address this issue, we propose a fast algorithm termed D-MUSIC, based on a measurement decoupling strategy. We demonstrate theoretically that for point sources with a known cluster structure, their measurement can be decoupled into local measurements of each of the clusters by solving a system of linear equations that are obtained by using multipole basis. We further develop a subsampled MUSIC algorithm to detect the cluster structure and utilize it to decouple the global measurement. In the end, MUSIC algorithm was applied to each local measurement to resolve point sources therein. Compared to the standard MUSIC algorithm, the proposed algorithm has comparable super-resolving capability while having a much lower computational complexity.

\section{Introduction}
In recent years, with the rapid development of novel imaging techniques, super-resolution is drawing increasing interest in the fields of imaging, signal processing, and applied mathematics. In this paper, we consider the super-resolution problem of resolving closely spaced point sources in one dimension from their noisy Fourier data in a bounded domain. 
The goal is to develop an efficient algorithm for the case of point sources with multi-cluster structure defined below. See also Figure \ref{fig: paperissueexam1} for a typical example. 
\begin{defi}(multi-cluster structure)\label{defi:clusterstructure}\\
	Let $\Omega$ be a cutoff frequency and 
	$D$ be a constant of order one. 
	Let $\Lambda_j$, $j=1, 2, \cdots, K$ be $K$ intervals centered at $O_j$ with half-length $\frac{D_j}{\Omega}$, respectively. 
	We say that
	$\cup_{j=1}^K \Lambda_j$
	is a $(K, L, D, \Omega)$-region if 
	\[
	\max_{1\leq j\leq K} D_j\leq D, \quad 
	\min_{1\leq i<j \leq K} |O_i-O_j| \geq \frac{L}{\Omega}.
	\] 
	We call that a set of point sources represented by a discrete measure $\mu$ has a multi-cluster structure if it is supported in a $(K, L, D, \Omega)$-region with $D \ll L$. 
\end{defi}  

Throughout the paper, we consider the following discrete measure with multi-cluster structure
\begin{equation} \label{eq-mu}
	\mu =\sum_{j=1}^{K}\sum_{q=1}^{n_j}a_{q,j}\delta_{y_{q,j}},
\end{equation}
where $y_{1,j},\cdots,y_{n_j,j}$ are the locations of point sources in $\Lambda_j$ and $a_{1,j},\cdots,a_{n_j,j}$ their amplitudes.
We denote
\[
m_{\min,j}=\min_{q=1,\cdots,n_j}|a_{q,j}|,
\quad m=\|\mu\|_{TV}=\sum_{j=1}^{K}\sum_{q=1}^{n_j}|a_{q, j}|.
\] 
We assume that we have the Fourier transform of $\mu$ restricted to a bounded interval $[-\Omega, \Omega]$, where $\Omega$ is the cutoff frequency. We sample at $N$ evenly-spaced points in the Fourier space to get the following discrete measurement
\begin{equation}\label{equ:fremeasurement1}
	\vect Y (x_l) = \mathcal{F}[\mu](x_l) + \mathbf W(x_l) = \sum_{j=1}^K\vect Y_j(x_l)+\vect W(x_l), \quad x_l \in [-1, 1], \ l=1,\cdots,N,
\end{equation}
where 
$$
\vect Y_j(x_l)= \sum_{q=1}^{n_j} a_{q,j}e^{iy_{q,j}\Omega x_l}, \quad
\mathcal{F}[\mu](x):= \int_{-\infty}^{\infty} \mu(y) e^{i\Omega y x}dy,
$$
and $\mathbf W(x_t)$ is the additive noise. Throughout the paper, we call $\vect Y$ the global measurement and 
$\vect Y_j$
the local measurement generated by point sources in the cluster $\Lambda_j$. 
We write (\ref{equ:fremeasurement1}) into the following vector form
\begin{equation}\label{equ:fremeasurement2}
	\vect Y = [\mu]+\vect W = \sum_{j=1}^K \vect Y_j +\vect W
\end{equation}
where $\vect Y = (\vect Y (x_1), \cdots, \vect Y (x_N))^T, [\mu] = (\mathcal{F}[\mu](x_1), \cdots, \mathcal{F}[\mu](x_N))^T, \vect W = (\mathbf W(x_1), \cdots, \mathbf W(x_N))^T$ and $\vect Y_j = (\vect Y_j(x_1), \cdots, \vect Y_j(x_N))^T$. 
We assume 
$$\frac{1}{\sqrt{N}}||\vect W||_2\leq \sigma.$$

We are interested in the inverse problem of recovering the source locations $y_{1,1}, \cdots, y_{n_K, K}$ from the measurement $\vect Y$. We note that 
by performing an inverse Fourier transform, the continuous measurement $\vect Y(x)$ is equivalent to the following spatial domain data 
$$
\hat{\vect Y} (t) =  \sum_{j=1}^{K}\sum_{q=1}^{n_j}a_{q,j} \frac{\sin{\Omega(t-y_{q,j}) }}{\pi t}. 
$$
Then the inverse problem becomes a deconvolution problem in the spatial domain. 

The inverse problem of (\ref{equ:fremeasurement2}) is severely ill-posed when there are multiple point sources closely spaced in one cluster, say with separation distance smaller than the Rayleigh length (RL) $\frac{\pi}{\Omega}$. 
For the case of a single cluster with closed spaced point sources, we refer the reader to \cite{liu2021theorylse, liu2021mathematicaloned, liu2021mathematicalhighd} for a theory of computational resolution limit which addresses the theoretical issues of recovering source number and locations. We note that fine structures with scale smaller than RL cannot be resolved by the usual sparsity promoting optimization methods \cite{candes2014towards, candes2013super, bernstein2019deconvolution, chi2020harnessing, cai2019fast, duval2015exact, 3cd521e14db44a81a63d9d80c4aaa1ca, tang2014near, morgenshtern2020super, denoyelle2017support, morgenshtern2016super}, which typically require a separation distance of several RLs (or some conditions to regularize the problem) \cite{duval2015exact, tang2015resolution}. 
Alternatively one may consider the classical subspace methods (e.g., MUSIC \cite{schmidt1986multiple}, Matrix Pencil \cite{hua1990matrix} and ESPRIT methods \cite{roy1989esprit}), which have demonstrated super-resolving capacity \cite{batenkov2019super, li2021stable, liao2016music, li2019super}. However, 
in the case of multiple clusters, highly sampled measurement is needed and this 
results large sized data matrix which demands high computational cost since these subspace algorithms rely on singular value decomposition.

\begin{figure}
	\centering
	\includegraphics[width=6in, height = 3.5in]{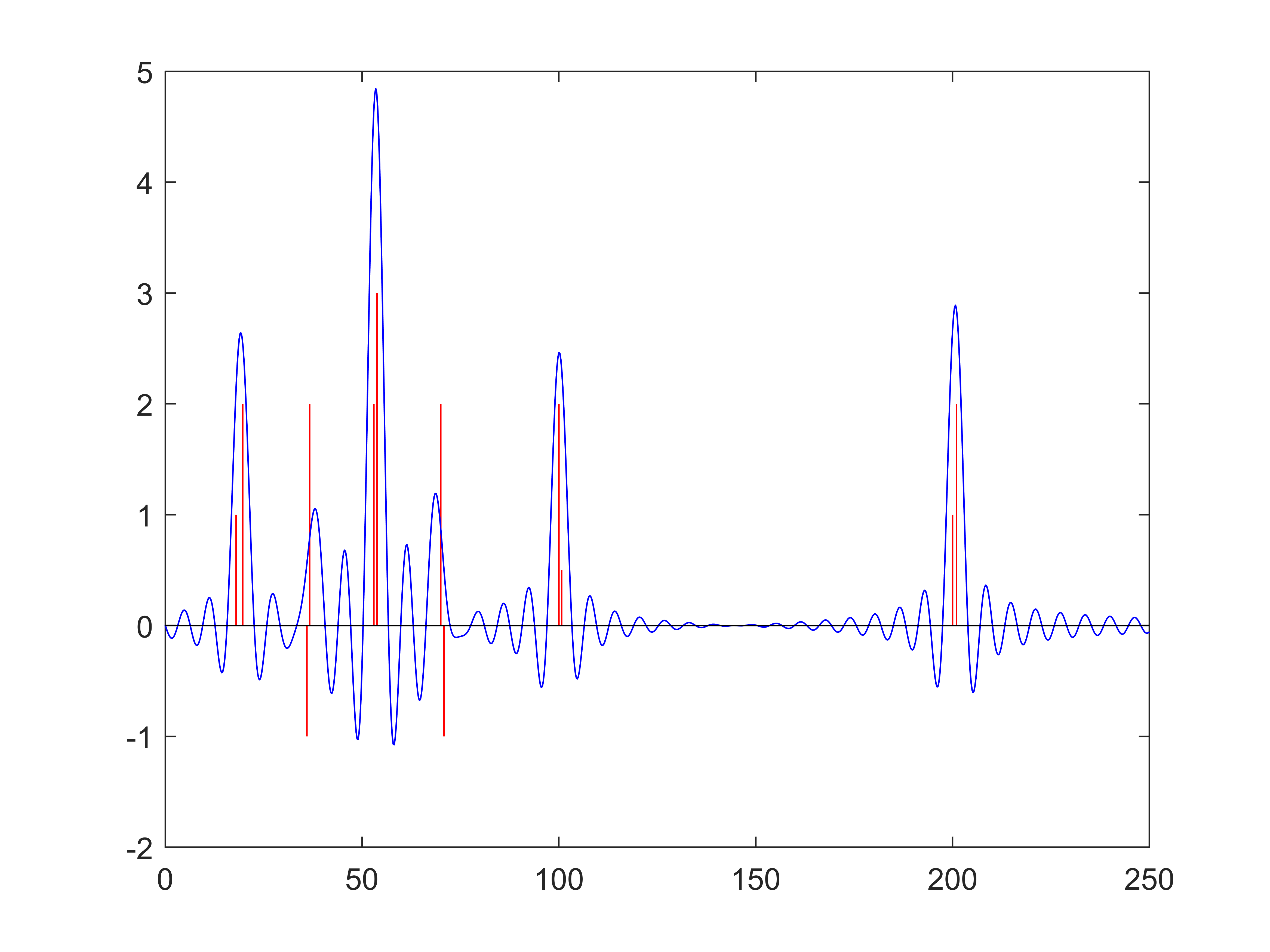}
	\caption{The blue curve represents spatial domain measurement; the red lines represent the to-be recovered point sources which are distributed in multiple clusters.}
	\label{fig: paperissueexam1}
\end{figure}

To remedy these issues, we develop an efficient algorithm based on a measurement decoupling strategy. The main idea is that for a given measurement $\vect Y$ that is generated by point sources as in (\ref{eq-mu}) with a known cluster structure, one can first recover local measurement $\vect Y_j$ for each $j$ and then 
resolve sources in each cluster from their local measurement. We prove that when the clusters are well-separated, global measurement can be decomposed into local measurements by solving a system of linear equations that are obtained by using multipole basis. 
We further develop a subsampled MUSIC algorithm to detect the cluster structure and utilize the result to decouple global measurement. A MUSIC algorithm was then applied to resolve point sources in each cluster from local measurement. It is demonstrated that the algorithm can super-resolve point sources when the clusters are well separated. Moreover, it has much lower computational complexity compared to standard MUSIC.

We notice that the idea of measurement decoupling was also exploited in \cite{wang2008tree}. Therein, the authors developed an algorithm for the two-dimensional DOA problem in array processing. They utilized a projection strategy to decouple the measurements of different groups in one dimension. Our measurement decoupling strategy is different. We use subsampled MUSIC algorithm to detect the cluster structure and then perform multipole expansion around each of the clustered centers. The local measurements are reconstructed using multipole basis.

On the other hand, if we assume that point sources are located on a grid, then the multi-cluster structure considered in this paper is related to block-sparse signals, see for instance \cite{stojnic2009reconstruction, eldar2009robust, eldar2009block, eldar2010block}. In \cite{stojnic2009reconstruction, eldar2010block}, the authors considered recovering $\vect X$ from $\vect Y=\vect D \vect X$, where  
\[
\vect X= (\underbrace{x_1, \cdots, x_d}_{\vect X^T[1]}, \underbrace{x_{d+1}, \cdots, x_{2d}}_{\vect X^T[2]}, \cdots, \underbrace{x_{N-d+1}, \cdots, x_{N}}_{\vect X^T[M]})^T
\]
is a concatenation of M blocks of vectors of length $D$, and 
\[
\vect D = (\underbrace{\vect D_1, \cdots, \vect D_d}_{\vect D[1]}, \underbrace{\vect D_{d+1}, \cdots, \vect D_{2d}}_{\vect D[2]}, \cdots, \underbrace{\vect D_{N-d+1}, \cdots, \vect D_{N}}_{\vect D[M]})
\]
is a concatenation of M blocks of matrice of size $L\times d$. 
The authors proposed the following relaxation scheme to reconstruct $\vect X$ \cite{stojnic2009reconstruction}:
\begin{equation}\label{equ:mixl2l1recovery}
	\min\sum_{l=1}^M ||\vect X[l]||_2, \ \text{subject to $\vect D\vect X=\vect Y$}.
\end{equation}
They demonstrated that when the block matrices $\vect D[j]$'s are Gaussian, (\ref{equ:mixl2l1recovery}) can find the sparest solution $\vect X$ with overwhelming probability as $N\rightarrow \infty$ under certain conditions. In \cite{eldar2009block}, based on block-coherence measure, 
the authors showed that any block $k$-sparse vector can be recovered
if the block-coherence satisfies certain condition. However, this condition does not hold for the measurement matrix when resolving closely spaced point sources as is considered in this paper.

The rest of the paper is organized in the following way. In Section 2, we introduce the theory and strategy for measurement decoupling. In Section 3, we develop a subsampled MUSIC algorithm to detect cluster structures. In Section 4, we develop the measurement decoupling based algorithm, D-MUSIC, for resolving point sources with multi-cluster structure. We also conduct numerical experiments to demonstrate its efficiency. Finally, in Section 5, we outline some future works.

\section{The theory of measurement decoupling by using multipole basis}\label{section:measuredecouptheory}
In this section, we develop the theory of measurement decoupling using multipole basis. 
The aim is to decouple local measurements $\vect Y_j$'s from the global measurement $\vect Y$ in (\ref{equ:fremeasurement1}) for point sources (\ref{eq-mu}) with a known multi-cluster structure. The main idea is to first represent each of the local measurement $\vect Y_j$ using a proper set of multipole basis, and then decouple them by solving a system of linear equations.
We show that the strategy works when the clusters are well-separated. In section \ref{sec:multipole}, we develop the theory using multipole basis in a straightforward manner. In section \ref{sec:multipole2}, we improve the decoupling strategy by using a modulation technique. 

\subsection{Decoupling using multipole basis-a precursor} \label{sec:multipole}
We start with the following multipole expansion for the local measurement $\vect Y_j$,
\begin{equation}\label{equ:multipoleexpanderive}
	\begin{aligned}
		\vect Y_j(x) = &\sum_{q=1}^{n_j} a_{q,j}e^{i\Omega y_{q,j}x} = \sum_{q=1}^{n_j}a_{q,j}e^{i\Omega O_jx}e^{i\Omega (y_{q,j}-O_j)x}\\
		=& \sum_{q=1}^{n_j}a_{q,j} \sum_{r=0}^{\infty}e^{i\Omega O_jx}\frac{(i\Omega (y_{q,j}-O_j)x)^r}{r!}\\
		= & \sum_{r=0}^{\infty}\sum_{q=1}^{n_j}a_{q,j}(\Omega(y_{q,j}-O_j))^r e^{i \Omega O_j x}\frac{(ix)^r}{r!}. 
	\end{aligned}
\end{equation}

We call the function $h_{r,O_j}(x)=\sqrt{2r+1}e^{i \Omega O_j x}(ix)^r$ the $r$-th order multipole function centered at $O_j$. 
Let $x_1, x_2, \cdots, x_N$ be $N$ equally spaced sample point in $[-1, 1]$. We define its discretized version by 
\begin{equation}\label{equ:multipolevector}
	\vect h_{r,O_j} = (h_{r,O_j}(x_1), \cdots, h_{r,O_j}(x_N))^T, 
\end{equation}
and call it the $r$-th order multipole basis vector centered at $O_j$. 
We have 
\begin{equation}\label{equ:normoffremultipole}
	\frac{1}{\sqrt{N}}||\vect h_{r, O_j}||_2\approx \sqrt{\frac{1}{2}\int_{-1}^1(2r+1)x^{2r}dx}\leq  1, \quad 1\leq j\leq K, \ r=0,1,\cdots.
\end{equation}
We define
\begin{equation}\label{equ:multipolecoeffi}
	Q_{r, O_j}(\mu) = \frac{\sum_{q=1}^{n_j}a_{q,j}(\Omega(y_{q,j}-O_j))^r}{\sqrt{2r+1} r!}
\end{equation}
and call it multipole coefficient. Using multiple basis vectors, we have the following representation for the gobal measurement $\vect Y$:
\begin{align}\label{equ:multipoleexpanequ1}
	\vect Y=\sum_{j=1}^{K} \sum_{r=0}^{\infty} Q_{r, O_j}\vect h_{r, O_j}+\vect W.
\end{align}

Observing that higher order multipole basis vectors decay exponentially fast as the order increases, we thus can approximate $\vect Y$ by the first $s$ multipole basis vectors in each of the clusters, where $s$ is to be determined. For the purpose, denote 
\begin{equation}\label{equ:multipolematrixsingle}
	\vect H[j]=\Big(\mathbf h_{0,O_j},\cdots,\mathbf h_{s-1,O_j}\Big)
\end{equation}
by the multipole matrix associated with the cluster centered at $O_j$, and 
$$
\vect \theta_j = (Q_{0,O_j},\cdots,  Q_{s-1,O_j})
$$
by the vector of multipole coefficients. We have  
\begin{equation}\label{equ:multipoleexpanequ2}
	\vect Y = \sum_{j=1}^K\vect H[j] \vect \theta_j+\vect W +\vect {Res}, 
\end{equation}
where $\mathbf {Res}$ is the residual term. Now, we determine $s$. For a given $\sigma, m$, and $D$, we choose 
\begin{equation}\label{equ:polesrecovered2}
	s:=\min\Big\{l\in \mathbb N: \frac{D^l(l+1)}{l!\sqrt{2l+1}(l+1-D)}\leq \frac{\sigma}{m}, \ l\geq D \Big\}.
\end{equation}
Then 
\begin{equation}\label{equ:multipoleresbound1}
	\begin{aligned}
		&\frac{1}{\sqrt{N}}||\mathbf {Res}||_2=\frac{1}{\sqrt{N}}\btwonorm{\sum_{r=s}^{+\infty}\sum_{j=1}^{K}\sum_{q=1}^{n_j}\frac{a_{q,j}(\Omega (y_{q,j}-O_j))^r}{r!\sqrt{2r+1}}\mathbf h_{r, O_j}}\\
		\leq&\sum_{r=s}^{+\infty}\sum_{j=1}^{K}\sum_{q=1}^{n_j}\frac{|a_{q,j}||\Omega (y_{q,j}-O_j)|^r}{r!\sqrt{2r+1}}\frac{1}{\sqrt{N}}||\mathbf h_{r, O_j}||_2
		\lesssim  m\sum_{r=s}^{+\infty}\frac{D^r}{r!\sqrt{2r+1}} \qquad \Big(\text{by (\ref{equ:normoffremultipole})}\Big)\\
		=& m\frac{D^s}{s!\sqrt{2s+1}}\Big(1+\frac{D}{s+1}+\frac{D^2}{(s+1)(s+2)}+\cdots\Big) \\
		<& m\frac{D^s}{s!\sqrt{2s+1}}\frac{1}{1-\frac{D}{s+1}}= m\frac{D^s(s+1)}{s!\sqrt{2s+1}(s+1-D)}\leq\sigma. \qquad \qquad \Big(\text{by (\ref{equ:polesrecovered2})}\Big)
	\end{aligned}
\end{equation}
Therefore, $\vect {Res}$ is of noise level with $s$ chosen above. 




We next present a result on the correlation between vectors spanned by mulitipole basis vectors associated with different clusters. 
\begin{lem}\label{lem:multipolemxtrixorthogonality1}
	Let $\vect H[j], \vect H[p]$ be defined as in  (\ref{equ:multipolematrixsingle}). Assume that $|O_j-O_p | \geq \frac{2(2s-2)^2}{\Omega}$, then
	\[
	\lim_{N\rightarrow \infty} \Big|\Big \langle \frac{1}{\sqrt{N}} \vect H[j]\vect a_j, \frac{1}{\sqrt{N}} \vect H[p] \vect a_p \Big\rangle\Big| < \frac{3.2(2s-1)^2}{|\Omega(O_j-O_p)|} \lim_{N\rightarrow \infty} \Big|\Big| \frac{1}{\sqrt{N}} \vect H[j]\vect a_j\Big|\Big|_2  \Big|\Big|\frac{1}{\sqrt{N}}\vect H[p] \vect a_p \Big|\Big|_{2}.  
	\]
\end{lem}
\begin{proof} Note that
	\begin{align*}
		&\lim_{N\rightarrow \infty} \Big|\Big\langle \frac{1}{\sqrt{N}} \vect H[j]\vect a_j, \frac{1}{\sqrt{N}} \vect H[p] \vect a_p \Big \rangle\Big| \\
		=& \Big|\int_{-1}^1 e^{i\Omega (O_j-O_p)x} \sum_{q_1, q_2=0}^{s-1} a_{q_1,j} \bar a_{q_2, p} i^{q_1}(-i)^{q_2} \sqrt{2q_1+1}\sqrt{2q_2+1} x^{q_1+q_2}dx\Big|.
	\end{align*}	
	Let $\psi(x) = \sum_{q_1, q_2=0}^{s-1} a_{q_1,j} \bar a_{q_2, p} i^{q_1}(-i)^{q_2} \sqrt{2q_1+1}\sqrt{2q_2+1} x^{q_1+q_2}$. By Lemmas \ref{lem:oscillatoryintegral} and \ref{lem:polynomialineq2} we have 
	\begin{align*}
		\Big|\int_{-1}^{1}e^{i \Omega (O_j-O_p)x} \psi(x)dx\Big| < \frac{3.2||\psi||_{L_{\infty}([-1,1])}}{|\Omega (O_j-O_p)|}
		\leq  \frac{3.2(2s-1)^2||\psi||_{L_{1}([-1,1])}}{|\Omega (O_j-O_p)|}.
	\end{align*}
	Finally, by Cauchy–Schwarz inequality we have 
	\begin{align*}
		&\Big|\Big|\psi\Big|\Big|_{L_{1}([-1,1])}\leq \sqrt{\int_{-1}^1\Big|\sum_{q_1=0}^{s-1}a_{q_1,j}\sqrt{2q_1+1} x^{q_1}\Big|^2dx} \sqrt{\int_{-1}^1\Big|\sum_{q_2=0}^{s-1}a_{q_2,j}\sqrt{2q_2+1} x^{q_2}\Big|^2dx}\\
		= &\lim_{N\rightarrow \infty}\Big|\Big| \frac{1}{\sqrt{N}} \vect H[j]\vect a_j\Big|\Big|_2  \Big|\Big|\frac{1}{\sqrt{N}}\vect H[p] \vect a_p \Big|\Big|_{2}.
	\end{align*}
	This completes the proof.
\end{proof}

\begin{remark}
	Lemmas \ref{lem:multipolemxtrixorthogonality1} demonstrates that when the clusters are well separated, the spaces spanned by multipole basis vectors associated with different clusters are nearly orthogonal to each other. In comparison, we note that in a recent paper \cite{batenkov2021spectral}, the authors obtained 
	similar result for the column vectors of Vandermonde matrix. 
	
	
\end{remark}

We are ready to present a first result on the decoupling theory. 
\begin{prop}\label{thm:measurementdecouple1}
	Suppose $N$ is large enough and that the point sources in (\ref{eq-mu}) is supported in a $(K, {L}, {D}, \Omega)$-region. Suppose 
	\begin{equation}\label{equ:clustersepacondition1}
		L\geq 12.8(2s-1)^2(\ln(\frac{K}{2})+1)   
	\end{equation}
	with $s$ being defined by (\ref{equ:polesrecovered2}). Let
	\begin{equation}\label{equ:meadecoupleastsquare1}
		(\vect a_1, \cdots, \vect a_K)=\text{argmin}_{\vect \theta_j} \frac{1}{\sqrt{N}}\Big|\Big|\sum_{j=1}^K\vect H[j] \vect \theta_j-\vect Y\Big|\Big|_2,
	\end{equation}
	then we have 
	\[
	\frac{1}{\sqrt{N}}\Big|\Big|\vect H[j] \vect a_j -\vect Y_j\Big|\Big|_2\lesssim \sigma, \quad 1\leq j \leq K,
	\]
	where $\vect Y_j$'s are the local measurements in (\ref{equ:fremeasurement2}). 
\end{prop}
\begin{proof} Note that
	\begin{equation}\label{equ:globalmeaappro}
		\frac{1}{\sqrt{N}}\Big|\Big|\sum_{j=1}^K\vect Y_j-\vect Y\Big|\Big|_2\leq \sigma.
	\end{equation}
	For each local measurement $\vect Y_j$, similar to (\ref{equ:multipoleexpanequ2}), we have 
	\begin{equation}\label{equ:localmeasuremultipolexpan1}
		\vect Y_j = \vect H[j] \vect b_j +\vect {Res}_j,
	\end{equation}
	where $\vect b_j = (Q_{0,O_j}, \cdots,  Q_{s-1,O_j})^T$ with $Q_{r,O_j}=\sum_{q=1}^{n_j}\frac{a_{q,j}(\Omega (y_{q,j}-O_j))^r}{r!\sqrt{2r+1}}$, and $\vect {Res}_j$ is the residual term. Similar to (\ref{equ:multipoleresbound1}), we have $\frac{1}{\sqrt{N}}||\sum_{j=1}^K\vect {Res}_j||_2\lesssim \sigma$. Together with (\ref{equ:globalmeaappro}), it follows that \\ $\frac{1}{\sqrt{N}}\Big|\Big|\sum_{j=1}^K\vect H[j] \vect b_j-\vect Y\Big|\Big|_2 \lesssim  \sigma$. On the other hand, it is clear that $\frac{1}{\sqrt{N}}\Big|\Big|\sum_{j=1}^K\vect H[j] \vect a_j-\vect Y\Big|\Big|_2 \lesssim \sigma$. Therefore
	\begin{equation}\label{equ:measuredecouple1equ1}
		\frac{1}{\sqrt{N}}\Big|\Big|\sum_{j=1}^K\vect H[j] (\vect a_j- \vect b_j) \Big|\Big|_2  \lesssim  \sigma.
	\end{equation}
	We next estimate the approximation of the local measurements. Consider $\frac{1}{\sqrt{N}}\Big|\Big|\vect H[j] (\vect a_j- \vect b_j) \Big|\Big|_2, 1\leq j\leq K$ at first. Let $\vect v_j = \frac{1}{\sqrt{N}}\vect H[j] (\vect a_j- \vect b_j)$, for large enough $N$, we have
	\begin{equation}\label{equ:measuredecouple1equ2}
		\begin{aligned}
			&\frac{1}{N}\Big|\Big|\sum_{j=1}^K\vect H[j] (\vect a_j- \vect b_j) \Big|\Big|_2^2 = \Big|\Big|\sum_{j=1}^K\vect v_j \Big|\Big|_2^2 = \sum_{j=1}^K\Big|\Big|\vect v_j \Big|\Big|_2^2+\sum_{j=1}^K\sum_{p\neq j} \Big\langle \vect v_j, \vect v_p\Big \rangle \\
			\geq & \sum_{j=1}^K\Big|\Big|\vect v_j \Big|\Big|_2^2 - \sum_{j=1}^K\sum_{p\neq j}  \frac{3.2(2s-1)^2}{|\Omega(O_j-O_p)|} \Big|\Big|\vect v_j \Big|\Big|_2 \Big|\Big|\vect v_p \Big|\Big|_2 \quad \Big(\text{by Lemma \ref{lem:multipolemxtrixorthogonality1}}\Big)\\
			\geq& \sum_{j=1}^K\Big|\Big|\vect v_j \Big|\Big|_2^2 - \sum_{j=1}^K\sum_{p\neq j}  \frac{1.6(2s-1)^2}{|\Omega(O_j-O_p)|}\Big(\Big|\Big|\vect v_j \Big|\Big|_2^2+ \Big|\Big|\vect v_p \Big|\Big|_2^2\Big)\\
			\geq & \sum_{j=1}^K\Big(\Big|\Big|\vect v_j \Big|\Big|_2^2 - \sum_{p\neq j} \frac{3.2(2s-1)^2}{|\Omega(O_j-O_p)|}\Big|\Big|\vect v_j \Big|\Big|_2^2\Big)\\
			\geq & \frac{1}{2}\sum_{j=1}^K\Big|\Big|\vect v_j \Big|\Big|_2^2 = \sum_{j=1}^K\frac{1}{2N}\Big|\Big|\vect H[j] (\vect a_j- \vect b_j) \Big|\Big|_2^2,
		\end{aligned}
	\end{equation}
	where the last inequality is derived from the assumption (\ref{equ:clustersepacondition1}) and the inequality $\sum_{p=1, p\neq j}^K\frac{1}{|p-j|}< 2(\ln \frac{K}{2}+1)$. Using (\ref{equ:measuredecouple1equ1}), it follows that
	\[
	\frac{1}{\sqrt{N}}\Big|\Big|\vect H[j] (\vect a_j- \vect b_j) \Big|\Big|_2 \lesssim \sigma, \quad j=1,\cdots, K.
	\]
	Furthermore, by (\ref{equ:localmeasuremultipolexpan1}) and $\frac{1}{\sqrt{N}}||\vect {Res}_j||_2\lesssim \sigma$,  we have  
	\[
	\frac{1}{\sqrt{N}}\Big|\Big|\vect H[j] \vect a_j- \vect Y_j \Big|\Big|_2\lesssim \sigma, \quad j=1,\cdots, K,
	\]
	which completes the proof.
\end{proof}

We observe that in Proposition \ref{thm:measurementdecouple1} the minimum separation distance between clusters depends on the number of the clusters $K$. This is due to the slow decay of the correlation between vectors in the span of mulitipole basis vectors associated with different clusters with respect to the cluster separation distance (see Lemma \ref{lem:multipolemxtrixorthogonality1}). 
To remedy this issue, we employ a modulation technique. This is done in the next section.

\subsection{Measurement decoupling using multipole basis with modulation}\label{sec:multipole2}
In this section, we decouple global measurement using modulated multipole basis. The modulation is intended to reduce the correlation between the multipole basis vectors from different clusters. 
For ease of presentation, we 
consider the following modulation function throughout
\[
f^t(x) = 1-x^2, \quad  x\in [-1, 1].   
\]
Other smooth functions with support in $[-1, 1]$ can be used as a modulation function as well and the treatment is similar. 
Throughout the paper, the superscript $t$ indicates that quantity is associated with modulation. We define 
$$
h_{r,O_j}^t(x)=\sqrt{2r+1}e^{i \Omega O_j x}(ix)^rf^t(x)
$$
to be the $r$-th order modulated multipole function centered at $O_j$ and
\begin{equation}\label{equ:multipolevector-modu}
	\vect h_{r,O_j}^t = (h_{r,O_j}^t(x_1), \cdots, h_{r,O_j}^t(x_N))^T
\end{equation}
its discretized version. 
We similarly have 
\begin{equation*}\label{equ:normoffremultipole-modu}
	\frac{1}{\sqrt{N}}||\vect h_{r, O_j}^t||_2\lesssim  1, \quad 1\leq j\leq K, \ r=0,1,\cdots. 
\end{equation*}
We consider the modulated measurement
\begin{equation}\label{equ:modulatedmeasurement}
	\vect Y^t(x_l) = f^t(x_l)\mathcal{F}[\mu](x_l) + f^t(x_l)\mathbf W(x_l) = \sum_{j=1}^K\vect Y_j^t(x_l)+\vect W^t(x_l), \quad x_l \in [-1, 1], \ l=1,\cdots,N.
\end{equation}
Or equivalently 
\begin{align}\label{equ:multipoleexpanequ1}
	\vect Y^t=\sum_{j=1}^{K}\sum_{r=0}^{\infty} Q_{r, O_j}\vect h_{r, O_j}^t+\vect W^t, 
\end{align}
where $Q_{r, O_j}$ is the same as defined in (\ref{equ:multipolecoeffi}).
Define $s$ as in (\ref{equ:polesrecovered2}), and 
\begin{equation}\label{equ:multipolematrixsingle-modu}
	\vect H^t[j]=\Big(\mathbf h_{0,O_j}^t,\cdots,\mathbf h_{s-1,O_j}^t\Big),\quad \vect \theta_j = (Q_{0,O_j},\cdots,  Q_{s-1,O_j})^T.
\end{equation}
We have 
\begin{equation}\label{equ:multipoleexpanequ2}
	\vect Y^t = \sum_{j=1}^K\vect H^t[j] \vect \theta_j+\vect W^t +\vect {Res}^t.  
\end{equation}
We can also show that $\frac{1}{\sqrt{N}}||\mathbf {Res}^t||_2 \lesssim \sigma$. Therefore, $\sum_{j=1}^K\vect H^t[j] \vect \theta_j$ can approximate $\vect Y^t$ to the noise level. 

We next show that the decay of the correlation between vectors in the span of modulated multipole basis vectors associated with different clusters with respect to the cluster separation distance is indeed enhanced, in comparison to Lemma \ref{lem:multipolemxtrixorthogonality1}.  

\begin{lem}\label{lem:multipolemxtrixorthogonality2}
	For $\vect H^t[j], \vect H^t[p], p\neq j$ defined as in (\ref{equ:multipolematrixsingle-modu}), assume that $|O_j-O_p| \geq \frac{2(2s+2)^2}{\Omega}$. Then  
	\begin{align*}
		&\lim_{N\rightarrow \infty} \Big|\Big \langle \frac{1}{\sqrt{N}} \vect H^t[j]\vect a_j, \frac{1}{\sqrt{N}} \vect H^t[p] \vect a_p \Big\rangle\Big| \\
		<& \frac{0.8(2s+2)^4(2s+3)^2}{|\Omega(O_j-O_p)|^3}\lim_{N\rightarrow \infty} \Big|\Big| \frac{1}{\sqrt{N}} \vect H^t[j]\vect a_j\Big|\Big|_2  \Big|\Big|\frac{1}{\sqrt{N}}\vect H^t[p] \vect a_p \Big|\Big|_{2}.  
	\end{align*}
\end{lem}

\begin{proof}
	The proof is similar to that of Lemma \ref{lem:multipolemxtrixorthogonality2}. It utilizes the second estimate in Lemma \ref{lem:oscillatoryintegral} since $f(\pm 1) =0, f'(\pm 1) =0$ for $f(x) = 1-x^2$. 
\end{proof}



We have the following main result on the measurement decoupling using multipole basis with modulation. 

\begin{thm}\label{thm:measurementdecouple2}
	Suppose $N$ is large enough and the point sources in (\ref{eq-mu}) is supported in a $(K, {L}, {D}, {\Omega})$-region. Suppose 
	\begin{equation}\label{equ:clustersepacondition2}
		L\geq 4^{1/3}(2s+3)^2    
	\end{equation}
	with $s$ being defined by (\ref{equ:polesrecovered2}). Let
	\begin{equation}\label{equ:meadecoupleastsquare-modu}
		(\vect a_1, \cdots, \vect a_K)= {\arg\min}_{\vect \theta_j} \frac{1}{\sqrt{N}}\Big|\Big|\sum_{j=1}^K\vect H^t[j] \vect \theta_j-\vect Y^t\Big|\Big|_2. 
	\end{equation}
	We have 
	\begin{equation}\label{equ:localmeasuremultipolexpan2}
		\frac{1}{\sqrt{N}}\Big|\Big|\vect H^t[j] \vect a_j -\vect Y_j^t\Big|\Big|_2\lesssim \sigma   
	\end{equation}
	for each local measurement $\vect Y_j^t, 1\leq j\leq K$ in (\ref{equ:modulatedmeasurement}).
\end{thm}
\begin{proof} In the same fashion as the proof of Proposition \ref{thm:measurementdecouple1}, for multipole expansion 
	\[
	\vect Y_j^t = \vect H^t[j] \vect b_j +\vect {Res}_j^t,
	\]
	we can show that 
	\begin{equation}\label{equ:measuredecouple2equ1}
		\frac{1}{\sqrt{N}}\Big|\Big|\sum_{j=1}^K\vect H^t[j] (\vect a_j- \vect b_j) \Big|\Big|_2\lesssim \sigma.
	\end{equation}
	Let $\vect v_j = \frac{1}{\sqrt{N}}\vect H^t[j] (\vect a_j- \vect b_j)$, for large enough $N$, we have
	\begin{equation}\label{equ:measuredecouple2equ2}
		\begin{aligned}
			&\frac{1}{N}\Big|\Big|\sum_{j=1}^K\vect H^t[j] (\vect a_j- \vect b_j) \Big|\Big|_2^2 = \Big|\Big|\sum_{j=1}^K\vect v_j \Big|\Big|_2^2 = \sum_{j=1}^K\Big|\Big|\vect v_j \Big|\Big|_2^2+\sum_{j=1}^K\sum_{p\neq j} \Big\langle \vect v_j, \vect v_p\Big \rangle \\
			\geq & \sum_{j=1}^K\Big|\Big|\vect v_j \Big|\Big|_2^2 - \sum_{j=1}^K\sum_{p\neq j}   \frac{0.8(2s+2)^4(2s+3)^2}{|\Omega(O_j-O_p)|^3} \Big|\Big|\vect v_j \Big|\Big|_2 \Big|\Big|\vect v_p \Big|\Big|_2\\
			\geq& \sum_{j=1}^K\Big|\Big|\vect v_j \Big|\Big|_2^2 - \sum_{j=1}^K\sum_{p\neq j}  \frac{0.4(2s+2)^4(2s+3)^2}{|\Omega(O_j-O_p)|^3}\Big(\Big|\Big|\vect v_j \Big|\Big|_2^2+ \Big|\Big|\vect v_p \Big|\Big|_2^2\Big)\\
			\geq & \sum_{j=1}^K\Big(\Big|\Big|\vect v_j \Big|\Big|_2^2 - \sum_{p\neq j} \frac{0.8(2s+2)^4(2s+3)^2}{|\Omega(O_j-O_p)|^3}\Big|\Big|\vect v_j \Big|\Big|_2^2\Big)\\
			\geq & \frac{1}{2}\sum_{j=1}^K\Big|\Big|\vect v_j \Big|\Big|_2^2 = \sum_{j=1}^K\frac{1}{2N}\Big|\Big|\vect H^t[j] (\vect a_j- \vect b_j) \Big|\Big|_2^2,
		\end{aligned}
	\end{equation}
	where the last inequality follows from (\ref{equ:clustersepacondition2}) and the estimate that $\sum_{p=1, p\neq j}^K\frac{1}{|p-j|^3}<2.5$. Using (\ref{equ:measuredecouple2equ1}), we get
	\[
	\frac{1}{\sqrt{N}}\Big|\Big|\vect H^t[j] (\vect a_j- \vect b_j) \Big|\Big|_2 \lesssim \sigma, \quad j=1,\cdots, K.
	\]
	Therefore  
	\[
	\frac{1}{\sqrt{N}}\Big|\Big|\vect H^t[j] \vect a_j- \vect Y_j^t \Big|\Big|_2\lesssim \sigma, \quad j=1,\cdots, K,
	\]
	which completes the proof. 
\end{proof}

Theorem \ref{thm:measurementdecouple2} demonstrates that the global modulated measurement can be decoupled into local modulated measurements when the clusters are well-separated. Compared to Proposition \ref{thm:measurementdecouple1}, the required separation distance between clusters is reduced due to the modulation technique. There is an alternative explanation to this. Observe that the point spread function (the measurement data in the spatial domain corresponding to a single point source) corresponding to the modulated measurement is given by 
$$
\int_{-1}^{1} (1-x^2)e^{-itx} dx.  
$$
By Lemma \ref{lem:oscillatoryintegral}, it has a decay rate of $1/|t^3|$, which is faster than the un-modulated one which has a decay rate of $1/|t|$. Therefore, the local modulated measurements associated with different clusters are more decorrelated as their separation distance increases.

On the other hand, we note that local measurements can be reconstructed by dividing the modulated ones $\vect H^t[j] \vect a_j$'s point-wisely by the modulation function $f^t$. Note that $f^t\approx 0$ for $x \approx \pm 1$. Therefore, only frequency components that are away from $\pm 1$ can be reconstructed stably, and those near the end points have to be discarded from the modulation technique. It can be shown that by choosing modulation function $f^t$ that has higher order of degeneracy at the end points $\pm 1$, one can improve the performance, both theoretically and numerically, of the decoupling of global modulated measurement. 
However, the benefit is at the cost of losing frequency component near the end points in the recovered local measurements. It is an interesting and important question to choose the optimal modulation function in practice. We leave this for a future investigation.

\begin{remark}
	Sufficiently many multipole basis vectors are needed for the measurement decoupling strategy in the above theorem. The choice of $s$ satisfying  (\ref{equ:polesrecovered2}) cannot be improved. Numerical experiments show that when the number of required multipole basis vectors $s$ is not big enough, local measurements cannot be recovered successfully from (\ref{equ:meadecoupleastsquare-modu}) 
	even though the clusters are well-separated and the global measurement is approximated to noise level. 
\end{remark}


Finally, we note that in Proposition \ref{thm:measurementdecouple1} and Theorem \ref{thm:measurementdecouple2}, the lower bound of $L$  depends on  $s$ in a quadratic manner. However, this estimate may not be optimal. On the other hand, it is clear that the lower bound increases as $s$ increases. In the next section, we conduct numerical experiments to demonstrate this dependence relation.

\subsection{The minimum required $L$ for the measurement decoupling}
In this section, we numerically investigate the dependence of the minimum required separation distance between clusters $L$ on the multipole  number $s$ for the two decoupling strategies in section \ref{sec:multipole} and \ref{sec:multipole2}. It is demonstrated that the modulation technique can relax the condition on the separation distance between clusters.

We first consider the decoupling strategy using the modulation  technique in section \ref{sec:multipole2}. For simplicity, we set $\Omega = 1, m=1, \sigma = 10^{-3}$, and the number of samples $N=1000$. We investigate the minimum $L$ required for $s=3, \cdots, 29,$ separately. More precisely, for each $s\in \{3, \cdots, 29\}$, we construct several $D$'s satisfying (\ref{equ:polesrecovered2}) with $\sigma = 10^{-3}$, and perform at least $1000$ random experiments for each pair $(D,  L)$ with $L\in \{3\pi, 3.5\pi,4\pi, \cdots, 50\pi\}$ being the lower bound for cluster separation distance. We consider realizations of point sources in a $(K, L, D, \Omega=1)$-region  with cluster number $K$ and cluster centers chosen randomly. We decouple the measurement as in Theorem \ref{thm:measurementdecouple2}. We then 
construct the multipole basis $\vect H^t[j]$'s and approximate the global measurement $\vect Y^t$. The global measurement is considered being well approximated if 
\begin{equation}\label{equ:wellapproximatecriterion}
	\frac{1}{\sqrt{N}}\Big|\Big|\sum_{j=1}^K\vect H^t[j] \hat {\vect a}_j-\vect Y^t\Big|\Big|_2\leq 3 \sigma.
\end{equation}  
If this is the case, we further recover the multipole coefficients by using  (\ref{equ:meadecoupleastsquare-modu}) and 
construct the local measurements. The measurement decoupling is regarded as successful if $\frac{1}{\sqrt{N}} ||\vect H^t[j] \hat {\vect a}_j - \vect Y_j^t||_2 < 6 \sigma$ for all $j$,  and otherwise unsuccessful. 


We view that the measurement decoupling strategy is successful for point sources with multi-cluster structure determined by the pair $(D,  L)$ if the success ratio out of 1000 random experiments is greater than $99\%$. For each $s\in\{3, \cdots, 29\}$, we denote $L(s)$ as the minimum $L$ in all the corresponding pairs $(D,L)$'s for successful measurement decoupling. We summarize the relation of $s$ and the $L(s)$ in Tables \ref{fig:minisepadisformeasuredecoup}. We perform similar experiments to the decoupling strategy in section \ref{sec:multipole}.
To compare the two decoupling strategy, we plot the relation of required minimum separation distance of clusters with respect to the multipole number $s$ for both strategies in Figure \ref{fig:minisepadisformeasuredecoup}. It is shown that the technique of modulation can indeed relax the minimum required separation distance of clusters for stable measurement decoupling.


\begin{table}[h]
	\centering
	\begin{tabular}{|c|c|c|c|c|c|c|c|c|c|}
		\hline
		s & 3 & 4 & 5 & 6 & 7 & 8 & 9 & 10 & 11 \\
		\hline
		L(s) & 3$\pi$  & 4$\pi$ & 5$\pi$ & 6$\pi$ &   7$\pi$ & 8$\pi$ &  9$\pi$ & 10$\pi$ & 11$\pi$  \\
		\hline
		s & 12 & 13 & 14 & 15 &  16 & 17 & 18 & 19 & 20 \\
		\hline
		L(s) & 12$\pi$ & 13$\pi$ & 14$\pi$ & 15$\pi$  & 16 $\pi$ & 18 $\pi$ & 19.5$\pi$ & 21$\pi$ & 22.5$\pi$\\
		\hline
		s  & 21 & 22 &23 & 24 &25 &26 & 27 & 28 & 29 \\
		\hline
		L(s)  & 24$\pi$ & 26 $\pi$ &  27.5$\pi$  & 29.5$\pi$& 31.5 $\pi$ & 33 $\pi$  & 36.5 $\pi$ &  38.5$\pi$  & 40.5 $\pi$\\
		\hline
	\end{tabular}
	\caption{Minimum separation of clusters for stably decoupling modulated measurements}
	\label{table:minidisformeasdecoup_xsquare}
\end{table}

\begin{figure}[h]
	\centering
	\includegraphics[width= 3.5in, height = 2.5in ]{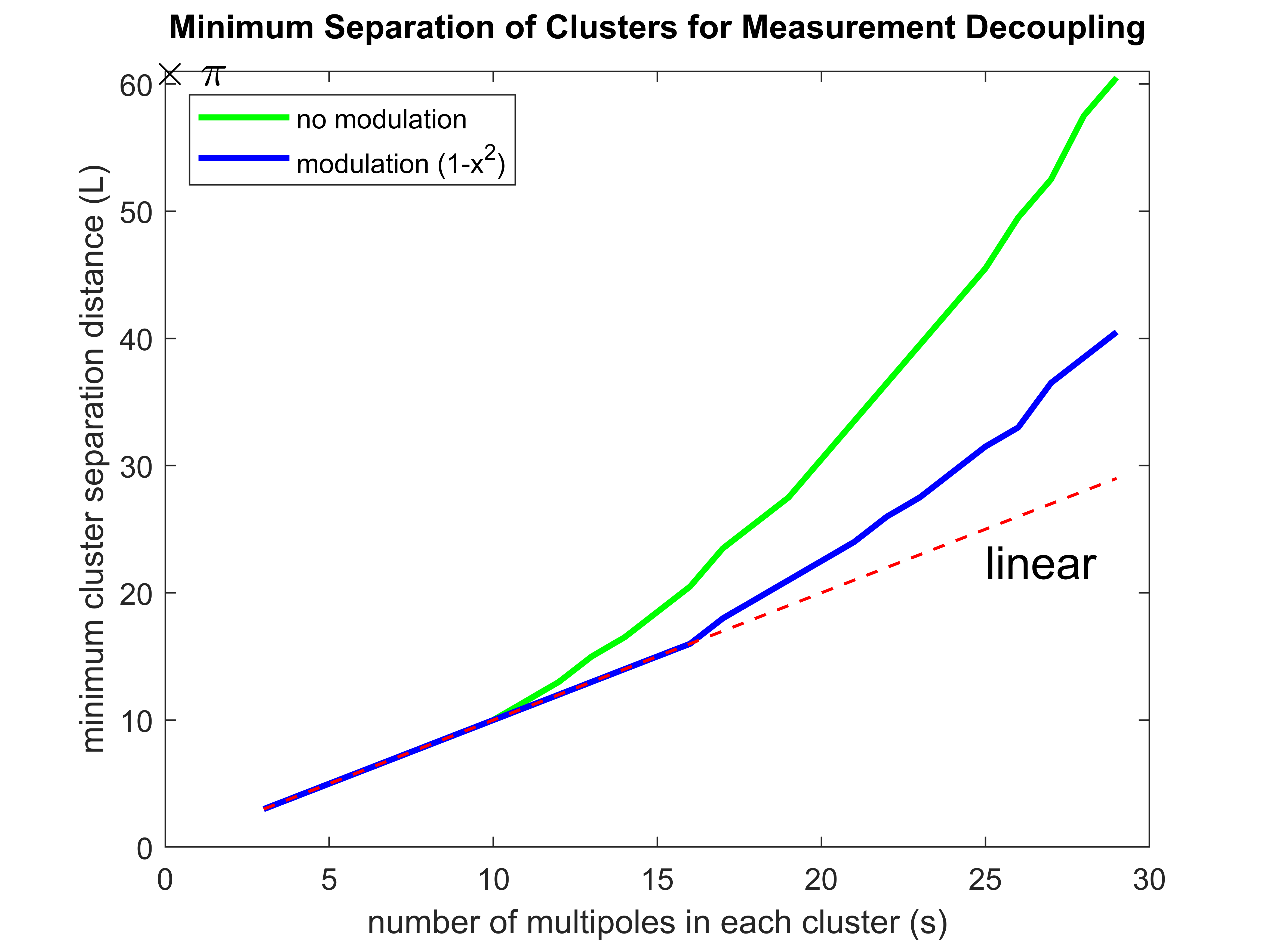}
	\caption{Plot of the minimum separation of clusters for measurement decoupling. It is shown that the required minimum separation distance between clusters can be relaxed by measurement modulation.}
	\label{fig:minisepadisformeasuredecoup}
\end{figure}



\section{A subsampled MUSIC algorithm for cluster structure detection}\label{section:clusterstructuredetection}
In this section, we develop a subsampled MUSIC algorithm to detect cluster structure for a given set of point sources with multi-cluster structure.

\subsection{MUSIC algorithm}\label{section:adaptivemusic}
We first review the standard MUSIC algorithm. We then incorporate a prior information on the cluster structure of point sources to make it more efficient. 
For simplicity, we set $\Omega =1$ in the subsequent presentation. 

In a standard MUSIC algorithm for solving the inverse problem (\ref{equ:fremeasurement1}), one first assemble the following Hankel matrix 
\begin{equation}\label{hankelmatrix1}\hat X=\begin{pmatrix}
		\mathbf Y(x_1)&\mathbf Y(x_2)&\cdots& \mathbf Y(x_{\hat N})\\
		\mathbf Y(x_2)&\mathbf Y(x_3)&\cdots&\mathbf Y(x_{\hat N+1})\\
		\cdots&\cdots&\ddots&\cdots\\
		\mathbf Y(x_{\hat N})&\mathbf Y(x_{\hat N+1})&\cdots&\mathbf Y(x_{2\hat N+1})
	\end{pmatrix},
\end{equation}
where $\hat N= \lfloor \frac{N-1}{2}\rfloor$. 
Then perform singular value decomposition for $\hat X$,
\[
\hat X=\hat U\hat \Sigma \hat U^*=[\hat U_1\quad \hat U_2]\text{diag}(\hat \sigma_1, \hat \sigma_2,\cdots,\hat \sigma_n,\hat \sigma_{n+1},\cdots,\hat \sigma_{\hat N+1})[\hat U_1\quad \hat U_2]^*,
\]
where $\hat U_1=(\hat U(1),\cdots,\hat U(n)), \hat U_2=(\hat U(n+1),\cdots,\hat U(\hat N+1))$ with $n$ being the estimated source number (model order). The source number $n$ can be detected by many algorithms such as those in \cite{wax1985detection, he2010detecting, han2013improved, liu2021theorylse}. Denote the orthogonal projection to the space $\hat U_2$ by $\hat P_2x=\hat U_2(\hat U_2^*x)$. For a test vector $\Phi(\omega)=(1, e^{ih\omega},\cdots,e^{i\hat Nh\omega})^T$ with $h$ being the spacing parameter, one define the MUSIC imaging functional 
\begin{align*}
	\hat J(\omega)=\frac{||\Phi(\omega)||_2}{||\hat P_2\Phi(\omega)||_2}=\frac{||\Phi(\omega)||_2}{||\hat U_2^*\Phi(\omega)||_2}.
\end{align*}
The local  maximizers of $\hat J(\omega)$ indicate the locations of the point sources. In practice, one can test evenly spaced points in a specified region and plot the discrete imaging functional and then determine the source locations by detecting the peaks. A peak selection algorithm is given in the appendix. 
We summarize the standard MUSIC algorithm in \textbf{Algorithm \ref{algo:standardmusic}} below.

\begin{algorithm}[H]
	\caption{\textbf{Standard MUSIC algorithm}}
	\textbf{Input:} Noise level $\sigma$, Measurements: $\mathbf{Y}=(\mathbf Y(x_1),\cdots, \mathbf Y(x_N))^T$ with $h$ the sampling distance\;
	\textbf{Input:} Region of test points $[TS, TE]$ and spacing of test points $TPS$\;
	1: Let $n$ be the estimated source number\;
	2: Let $\hat N =\lfloor\frac{N-1}{2}\rfloor$, formulate the $(\hat N +1)\times (\hat N +1)$ Hankel matrix $\hat X$ from $\mathbf{Y}$\;
	3: Compute the singular vector of $\hat X$ as $\hat U(1), \hat U(2),\cdots,\hat U(\hat N +1)$ and formulate the noise space $\hat U_{2}=(\hat U(n+1),\cdots,\hat U(\hat N +1))$\;
	4: For test points $\omega$'s in $[TS, TE]$ evenly spaced by $TPS$, construct the test vector $\Phi(\omega)=(1,e^{ih\omega}, \cdots, e^{i\hat N h\omega})^T$\; 
	5: Plot the MUSIC imaging functional $\hat J(\omega)=\frac{||\Phi(\omega)||_2}{||\hat U_2^*\Phi(\omega)||_2}$\;
	6: Select the peak locations $\hat y_j$'s in the $\hat J(\omega)$ by \textbf{Algorithm \ref{algo:peakselection}}\;
	\textbf{Return} $\hat y_j$'s.
	\label{algo:standardmusic}
\end{algorithm}

Now, assume that the point sources we are interested in are located in an interval $\Lambda = [O-D, O+D]$ that centered at $O$ with size $D$. 
We can incorporate this a prior information into the standard MUSIC to make it more efficient. To be more specific, let $y_1, \cdots, y_n$ be the point sources and its measurement without noise is given by
\begin{equation}\label{equ:localmeasurementeq1}
	\vect Y(x_l) = \sum_{q=1}^{n} a_{q}e^{iy_{q} x_l}, \ x_l \in [-1, 1], \ q=1, \cdots, n.
\end{equation}
Observe that 
\[
\vect Y(x_l) = e^{iOx_l}\sum_{q=1}^{n} a_{q}e^{i(y_{q}-O) x_l} =e^{iOx_l} {\vect Y^c}(x_l) , \ x_l \in [-1, 1],
\]
where $ {\vect Y^c}(x_l)= \sum_{q=1}^{n} a_{q}e^{i(y_{q}-O)x}$ is called the 
centralized local measurement. 
Note that the relative positions $\tilde{y}_{q}:= y_q-O$'s are located in $[-D, D]$. We can sample ${\mathbf Y^c}$ at $x_l\in[-1,1]$ with spacing $\frac{\pi}{2D}$ and use the samples to reconstruct $\tilde{y}_{q}$'s by the standard MUSIC algorithm. The original source locations can be further recovered as $\hat y_{1}=\tilde{y}_{1}+O,\cdots, \hat y_{n} = \tilde{y}_{n}+O$. We detail these steps in \textbf{Algorithm \ref{algo:adaptiveregionbasedmusic}} below.

\medskip
\begin{algorithm}[H]
	\caption{\textbf{MUSIC algorithm with a prior information}}	
	\textbf{Input:} Noise level $\sigma$, local measurements $\vect Y$\;
	\textbf{Input:} Cluster center $O$, cluster size $D$\;
	\textbf{Input:} Spacing of test points $TPS$\;
	1:  construct the centralized measurement ${\vect Y^c}(x_l) = \vect Y(x_l)e^{-i O x_l}$ with $x_l$'s spacing by $\frac{\pi}{2D}$\;
	2: Input $[-D, D], TPS, \sigma$, and ${\mathbf Y^c}$ into \textbf{Algorithm \ref{algo:standardmusic}} to recover the relative source locations $\tilde y_{1}, \cdots, \tilde y_{n}$\;
	3: Recover the source locations that $\hat y_{1}=\tilde y_{1}+O, \cdots, \hat y_{n}=\tilde y_{n}+O$\;
	\textbf{Return:} $\hat y_{q}$'s.
	\label{algo:adaptiveregionbasedmusic}
\end{algorithm}

We note that for point sources with multi-cluster structure considered in this paper, we can apply the above MUSIC algorithm to each of the local measurements. 


\subsection{Cluster structure detection}
In this section, we develop a subsampled MUSIC algorithm to detect cluster structures. 
We assume all the point sources are located in a $(K, L, D, \Omega)$-region with $\Omega=1$. We first choose an sufficiently large interval $[\tilde O-\tilde D, \tilde O+\tilde D]$ that covers all the sources. We then choose a proper shrinkage factor $0<\lambda<1$, and apply \textbf{Algorithm \ref{algo:adaptiveregionbasedmusic}} to the global measurement $\vect Y$ with samples in the interval $[-\lambda , \lambda]$ 
to get a set of point locations, say $c_j$, $1\leq j \leq K'$ for some integer $K'\geq K$. Note that due to subsampling, these locations are not necessarily the locations of the original point sources. However, their presence indicate that there are point sources nearby. 

We next estimate the cluster structures. We showed in \cite{liu2021theorylse, liu2021mathematicaloned} that $n$ point sources can be resolved if minimum separation distance between them is great than
\begin{equation}\label{equ:resolutionlimit}
	C\pi\Big(\frac{\sigma}{m_{\min}}\Big)^{\frac{1}{2n-1}},
\end{equation}
where $C$ is a constant. On the other hand, it is shown numerically that MUSIC algorithm can resolve the point sources under the above condition, see \cite{li2021stable}. Therefore, with a shrinkage factor $\lambda$ the resolution of MUSIC algorithm in the preceding step is of order $O\Big(\frac{\pi}{\lambda}(\frac{\sigma}{m_{\min}})^{\frac{1}{2n-1}}\Big)$. It indicates that when there are two point sources separated greater than $\frac{C\pi}{\lambda}(\frac{\sigma}{m_{\min}})^{\frac{1}{3}}$, the MUSIC algorithm will give two peaks. Thus for a peak centered at $c_j$, the point source/sources that correspond to it should be located in the interval $[c_j-\frac{C\pi}{\lambda}(\frac{\sigma}{m_{\min}})^{\frac{1}{3}}, c_j+\frac{C\pi}{\lambda}(\frac{\sigma}{m_{\min}})^{\frac{1}{3}}]$. After many numerical experiments, we choose $C=2$, i.e., the point sources are located in the interval $[c_j-\frac{2\pi}{\lambda}\sigma^{\frac{1}{3}}, c_j+\frac{2\pi}{\lambda}\sigma^{\frac{1}{3}}]$ for $m_{\min}\approx 1$.  

We denote $\Gamma_j = [c_j-d, c_j+d]$ with $d = \frac{2\pi}{\lambda \Omega}\sigma^{\frac{1}{3}}$. Notice that there may be multiple $c_j$'s reconstructed from the subsampled MUSIC algorithm that come from the same cluster. In this case, we may need to combine the involved intervals to get the right cluster structure. For the purpose, we introduce a parameter $ICT$ called the interval combining threshold. We combine the interval $\Gamma_j$'s if their centers has distance smaller than $ICT$. It is clear that $ICT$ should be a proper estimate of the full cluster size $2D$. On the other hand, to ensure that the cluster structure can be recovered successfully, one need the condition that the separation distance between clusters are much larger than their size, i.e. $L \gg D$. 
We summarize the cluster structure detection algorithm as \textbf{Algorithm \ref{algo:wholeclusterstructurealgorithm}} below.

\begin{algorithm}[H]
	\caption{\textbf{Cluster structure detection}}
	\textbf{Input:} Noise level $\sigma$, Measurement $\mathbf Y = (\mathbf Y (x_1), \cdots, \mathbf Y(x_N))$, Shrinkage factor: $\lambda$\;
	\textbf{Input:} Initial cluster center $\tilde O$, initial cluster size $\tilde D$\;
	\textbf{Input:} Spacing of source test points $TPS$\;
	\textbf{Input:} Interval combining threshold ICT\;
	1: Construct the subsampled measurement $\vect Y^s$ by deleting components $Y(x_l)$'s in $\vect Y$ with $|x_l|>\lambda $;\\
	2: Input $\tilde O, \tilde D, TPS, \sigma$, and $\vect Y^s$ into \textbf{Algorithm \ref{algo:adaptiveregionbasedmusic}} and recover the centers $c_j$'s\;
	3: Let $d = \frac{2\pi}{\lambda}\sigma^{\frac{1}{3}}$. 
	If the neighboring $c_j$'s are separated less than $ICT$, then the corresponding intervals are combined to form a new interval. After combining all the closely-spaced intervals, we can recover the cluster centers $O_j$'s as the centers of the new intervals and the cluster sizes $D_j$'s as the corresponding interval size\;
	4:Return cluster centers $O_j$'s, cluster size $D_j$'s. 
	\label{algo:wholeclusterstructurealgorithm}
\end{algorithm}



We remark that in the above cluster structure detection algorithm, the choice of $\lambda$ plays an important role. It depends on the cluster structure, the noise level and the available computational power. For large $\lambda$, say $\lambda \approx 1$, the underlying cluster structure can be definitely detected. However, it demands high computational cost. On the other hand, for small $\lambda$, the cost is reduced, however, the algorithm may not find the cluster structure.
In our numerical experiments, we choose $\lambda = \frac{1}{2}$.  
Note that one can choose a list of shrinkage factors $[\lambda_1, \cdots, \lambda_M]$ with $0< \lambda_1<\cdots<\lambda_M<1$ and continually detect the cluster structures for each shrinkage factor $\lambda_j$ until the cluster structure is detected. We leave the question of determining the optimal $\lambda$ for a future work.

\section{Measurement decoupling based super-resolution algorithm}\label{section:damusicalgorithm}
In this section, we develop a fast algorithm for super-resolving point sources with multi-cluster structure. It exploits the ideas of measurement decoupling and is termed D-MUSIC.



For a given set of point sources with multi-cluster structure, say $(K, D, L, \Omega)$, we first detect the cluster structure by \textbf{Algorithm \ref{algo:wholeclusterstructurealgorithm}} with a properly chosen shrinkage factor $\lambda$. 
We then decouple the global measurement into local measurements using the strategy in section \ref{sec:multipole2}. More precisely, we set $m=1$ and calculate the multipole number $s$ by (\ref{equ:polesrecovered2}). We then 
construct the multipole basis $\vect H^t[j]$'s and  
recover the multipole coefficients by 
\[
(\hat {\vect \theta}_1, \cdots,  \hat{\vect \theta}_K )= {\arg\min}_{\vect \theta_j, 1\leq j\leq K} \btwonorm{\sum_{j=1}^K\mathbf H^t[j] \vect \theta_j-\mathbf Y^t}. 
\]
The measurement decoupling is deemed successful if the 
residual term $\vect{Res}^t = \mathbf Y^t- \sum_{j=1}^K\mathbf H^t[j] \hat {\vect \theta}_j$ satisfies the condition
$$
\|\vect{Res}^t \|\leq C_{mea}\sigma
$$
for some constant $C_{mea}$.  Here instead of using 
$\vect H^t[j] \hat{\vect \theta}_j$'s for the modulated local measurements, we use the following data:
$$
\tilde {\vect Y}_j^t=\vect H^t[j] \hat{\vect \theta}_j+\vect{Res}^t, j=1, \cdots, K. 
$$
The reason is that numerically $\tilde {\vect Y}_j^t$ defined above leads to better reconstruction result when was fed to the MUSIC algorithm. See Figure \ref{fig:musicimgfromdiffmeasurments} for a numerical evidence. It is not clear what is reason behind such an interesting phenomenon. 
The local measurement can be reconstructed as
\[
\tilde {\vect Y}_j(x_l) = \tilde{ \vect Y}_j^t(x_l)/ (1-x_l^2),\ \text{for $x_l\in [-C_{msf}, C_{msf}]$},
\]
where $C_{msf}< 1$ is a cutoff threshold that ensure that $\tilde {\vect Y}_j(x_l)$ is reconstructed stably. The choice of $C_{msf}$ depends on the noise level and the behavior of the modulation function near the cutoff frequency. We summarize the detail of the measurement decoupling in \textbf{Algorithm \ref{algo:measurementdecouplebymultipole}} below. 

\begin{algorithm}[H]
	\caption{\textbf{Measurement decoupling by multipole expansion}}	
	\label{algo:measurementdecouplebymultipole}
	\textbf{Input:} Noise level $\sigma$, Measurement $\mathbf{Y}=(\mathbf Y(x_1),\cdots, \mathbf Y(x_N))^T$\;
	\textbf{Input:} Cluster centers $O_1,\cdots, O_K$, Cluster sizes $D_1, \cdots, D_K$\;
	\textbf{Input:} Noise tolerance factor $C_{mea}$, Measurement modulate function $f^t$, Measurement cutoff threshold $C_{msf}$\;
	1: Let $D=\max_{j=1}^K(D_j)$, compute the number of multipoles $s$ by $s\geq D, \frac{D^s (s+1)}{s!\sqrt{2s+1}(s+1-d)}\leq \sigma$\;
	2: Construct the modulated measurement $\vect Y^t(x_l) = \vect Y(x_l)f^t(x_l), \ l=1, \cdots, N$\;
	3: Construct the corresponding multipole matrix $\mathbf H^t[j] = (\mathbf h_{0,O_j}^t,\cdots,\mathbf h_{s-1,O_j}^t)$ by (\ref{equ:multipolematrixsingle})\;
	4: Recover $\hat {\vect \theta}_j$'s by ${\arg\min}_{\vect \theta_j, 1\leq j\leq K} ||\sum_{j=1}^K\mathbf H^t[j] \vect \theta_j-\mathbf Y^t||_2$ and the residual term is $\vect{Res}^t = \vect{Y}^t - \sum_{j=1}^K\mathbf H^t[j] \hat {\vect \theta}_j$\; 
	5: \If{ $\frac{1}{\sqrt{N}}||\sum_{j=1}^K\mathbf H^t[j] \hat {\vect \theta}_j -\mathbf Y^t||_2\leq C_{mea} \sigma$}
	{ 
		recover modulated local measurements $\tilde{\vect Y}_j^t = \vect H^t[j] \hat{\vect \theta}_j+ \vect {Res}^t, \ j=1, \cdots, K$\;
		recover local measurements $\tilde{\vect Y}_j(x_l) = \tilde{ \vect Y}_j^t(x_l)/f^t(x_l)$, for $x_l\in (-C_{msf}, C_{msf})$\;
		\textbf{Return:} measurements $\tilde {\vect Y}_j, \ j=1, \cdots, K$, DECOUPLE = SUCCESS.}
	\Else{
		\textbf{Return:} measurement $\vect Y$, DECOUPLE = FAIL.		
	}
\end{algorithm}

\begin{figure}[!h]
	\centering
	\begin{subfigure}[b]{0.28\textwidth}
		\centering
		\includegraphics[width=\textwidth]{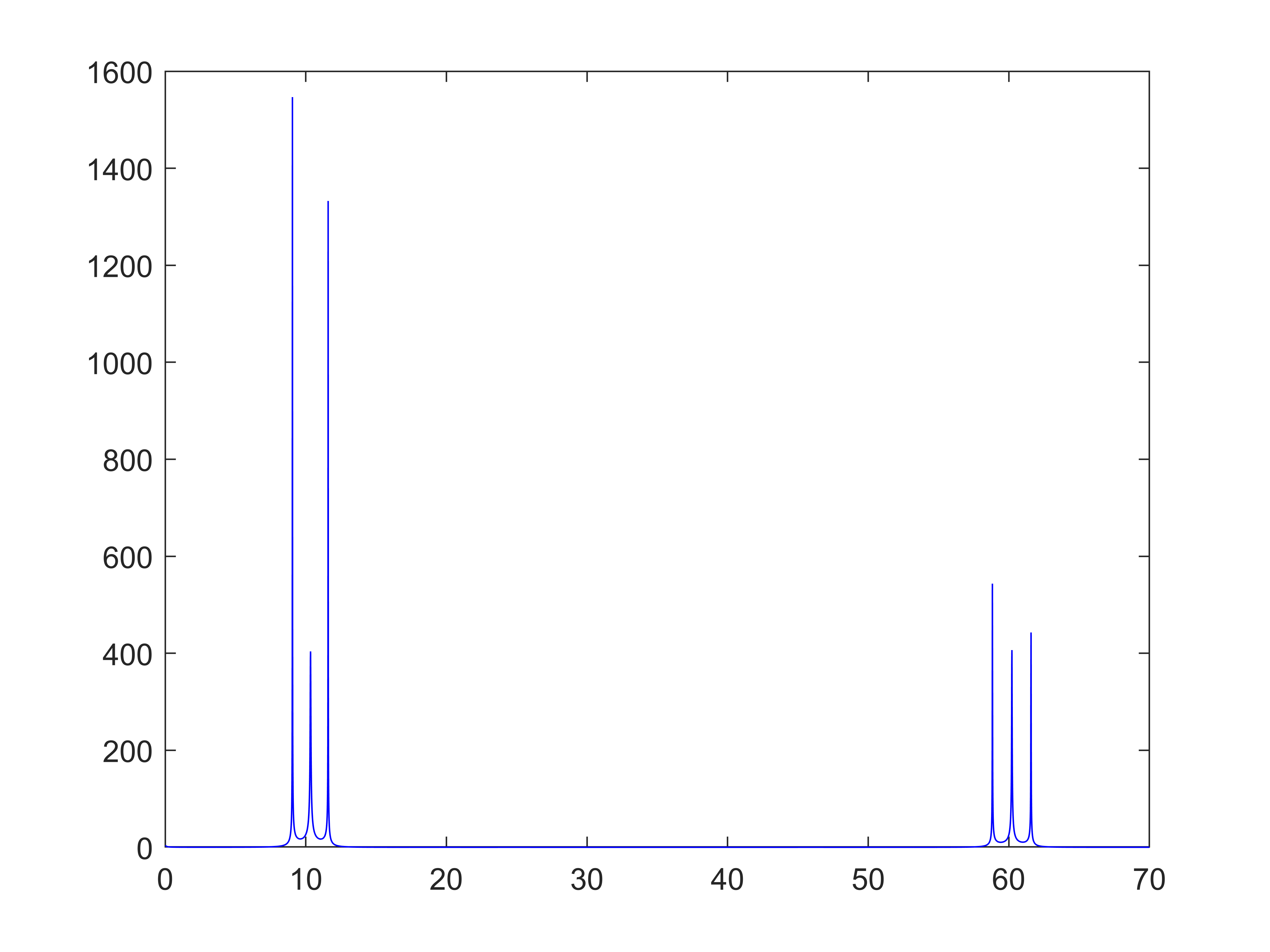}
		\caption{MUSIC image from measurement $\vect Y$}
	\end{subfigure}
	\begin{subfigure}[b]{0.28\textwidth}
		\centering
		\includegraphics[width=\textwidth]{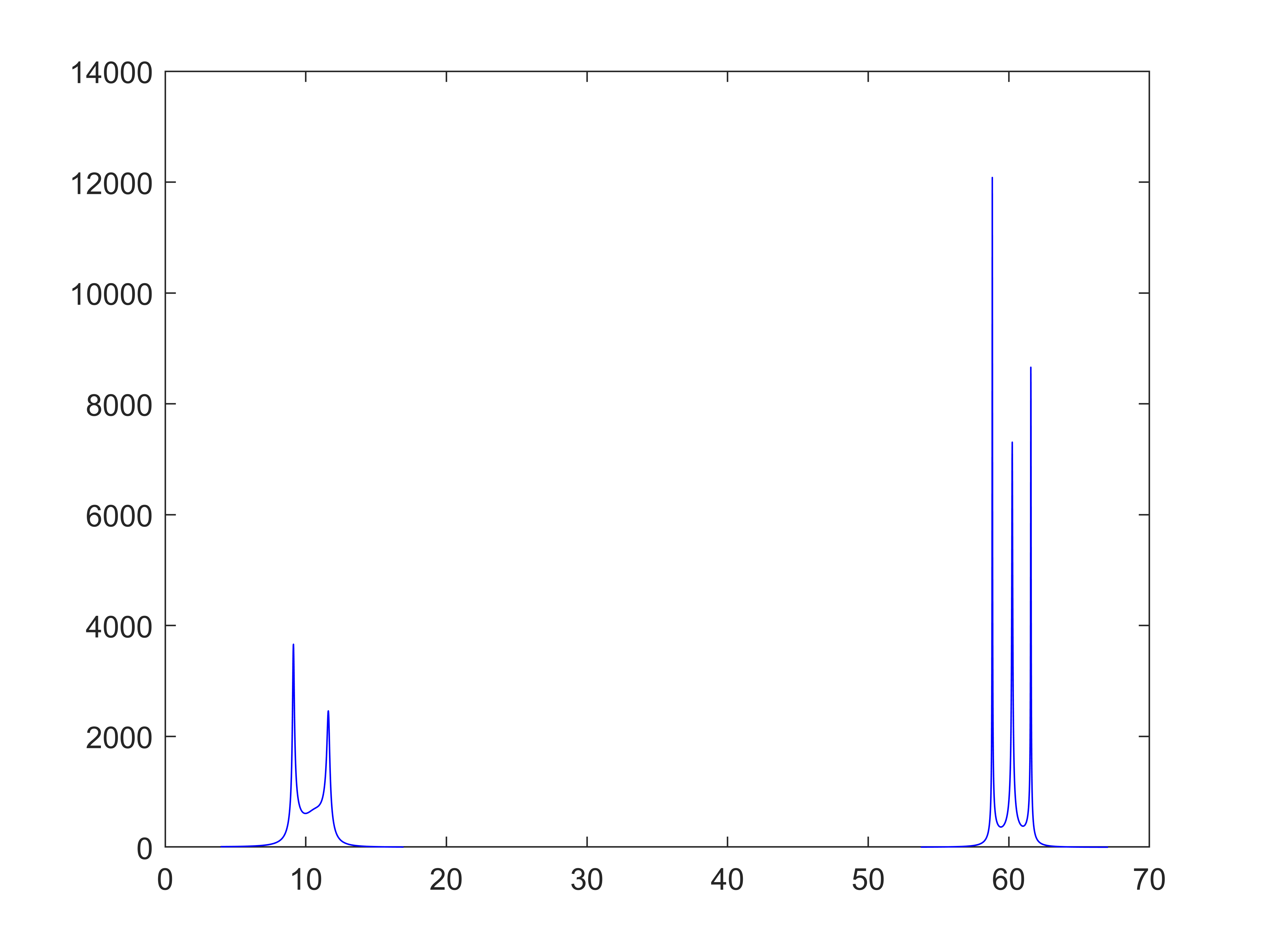}
		\caption{MUSIC image from $\hat {\vect Y}_j=\vect H[j] \hat{\vect \theta}_j$}
	\end{subfigure}
	\begin{subfigure}[b]{0.28\textwidth}
		\centering
		\includegraphics[width=\textwidth]{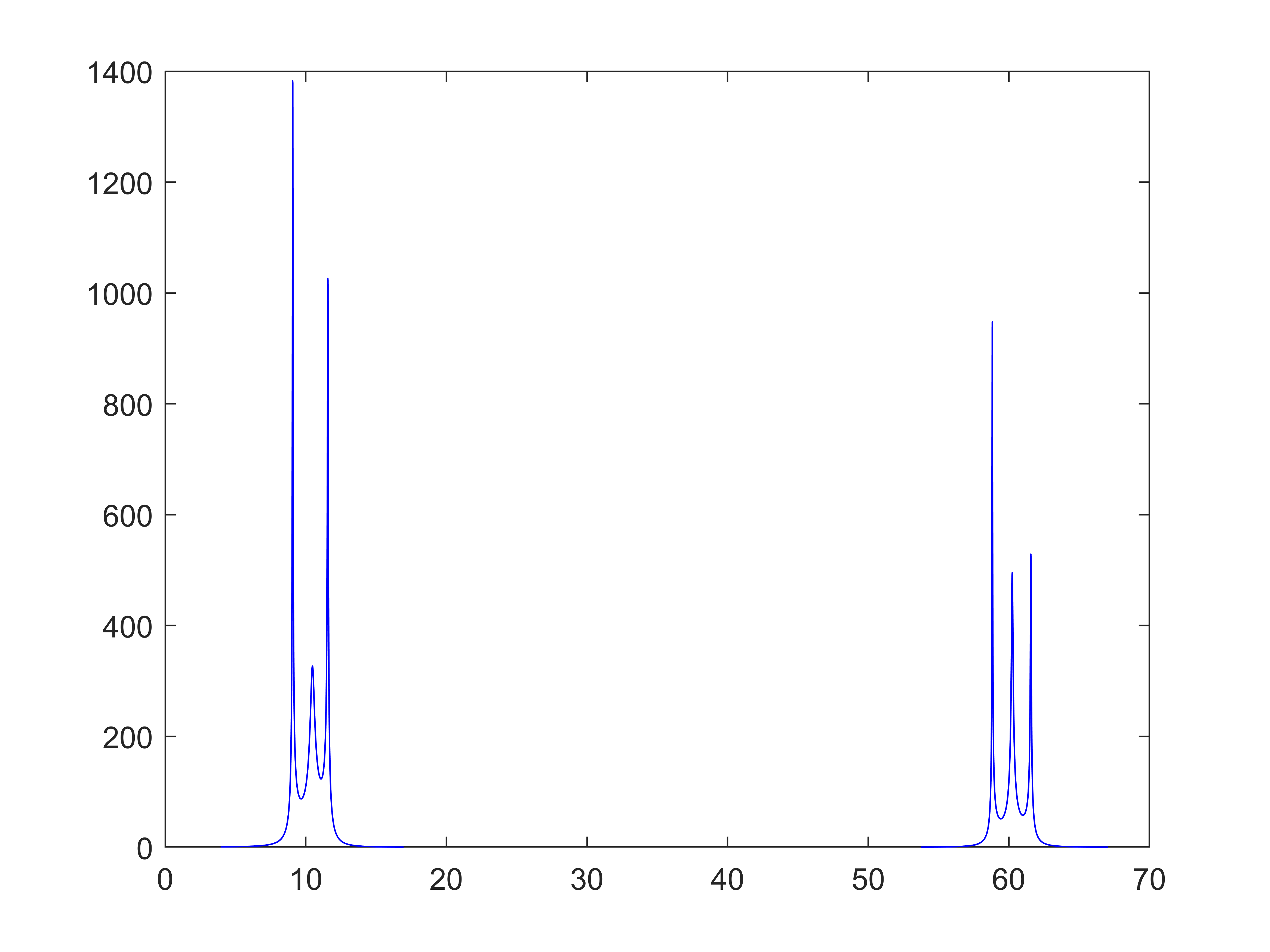}
		\caption{MUSIC image from $\hat {\vect Y}_j=\vect H[j] \hat{\vect \theta}_j+\vect{Res}$}
	\end{subfigure}
	\caption{Plots of the MUSIC images of different measurements (no measurements modulation). Figures (a), (b), and (c) are MUSIC images from measurements $\vect Y$, $\hat {\vect Y}_j=\vect H[j] \hat{\vect \theta}_j$, and $\hat {\vect Y}_j=\vect H[j] \hat{\vect \theta}_j+\vect{Res}$, respectively.  It is shown that three peaks in the left cluster cannot be recovered from the local measurements $\hat {\vect Y}_j=\vect H[j] \hat{\vect \theta}_j$, but can be restored from $\hat {\vect Y}_j=\vect H[j] \hat{\vect \theta}_j+\vect{Res}$. }
	\label{fig:musicimgfromdiffmeasurments}
\end{figure}

Finally, we recover the source locations from each local measurement by \textbf{Algorithm \ref{algo:adaptiveregionbasedmusic}}. We summarize the whole algorithm as \textbf{Algorithm \ref{algo:adaptivemeasurementdecouplingsupportrecovery}}.

\begin{algorithm}[H]
	\caption{\textbf{Decoupling based Adaptive MUSIC algorithm (D-MUSIC)}}
	\textbf{Input:} Noise level $\sigma$, Measurements: $\mathbf{Y}=(\mathbf Y(x_1),\cdots, \mathbf Y(x_N))^T$, Shrinkage factor: $\lambda$\;
	\textbf{Input:} Initial guess of the interval containing all sources $[\tilde O - \tilde{D}, \tilde O + \tilde{D}]$\;
	\textbf{Input:} Spacing of test points for cluster centers  $TPS_{cluster}$, Spacing of test points for point sources  $TS_{source}$, and Interval combine threshold $ICT$\;
	\textbf{Input:} measurement shrinkage factor $C_{msf}$\;
	1: Initialize the local measurement list as $LML = [\vect Y]$\;
	2: Initialize the cluster center list and the cluster size list as $CCL=[\tilde O], CSL=[\tilde D]$\;
	3: Input $\sigma, \lambda, \vect Y, \tilde O, \tilde D$ and $TPS_{cluster}$ to \textbf{Algorithm \ref{algo:wholeclusterstructurealgorithm}} to recover $K$ cluster centers $O_{1},\cdots, O_{K}$ and the cluster sizes $D_1, \cdots, D_K$\;
	4:Use \textbf{Algorithm \ref{algo:measurementdecouplebymultipole}} to recover the local measurements, ${\vect Y}_1, \cdots, {\vect Y}_K$\;
	5:\If {$DECOUPLE == SUCCESS$}{
		Update that $LML =[\vect Y_1, \cdots, \vect Y_K]$, $CCL = [O_1, \cdots, O_K]$, $CSL = [D_1, \cdots, D_K]$\;
	}
	6: Input each local measurement $\vect Y_j$, corresponding cluster center $O_j$, and cluster size $D_j\ $ in $LML, CCL, CSL$ respectively into \textbf{Algorithm \ref{algo:adaptiveregionbasedmusic}} to recover all the source locations $\hat y_j, j=1, \cdots, n$\;
	\textbf{Return:} LOCATIONS = $[ \hat y_1, \cdots, \hat y_n]$.
	\label{algo:adaptivemeasurementdecouplingsupportrecovery}
\end{algorithm}

\medskip
We now estimate the computational complexity of \textbf{Algorithm \ref{algo:adaptivemeasurementdecouplingsupportrecovery}}. Recall that $N$ is the number of total samples in the measurement. We first consider \textbf{Algorithm \ref{algo:wholeclusterstructurealgorithm}}. Due to subsampling, the number of samples is $\lambda N$ and the computational complexity of SVD therein is of order $O(\lambda^3N^3)$. In addition, the computational complexity of constructing MUSIC imaging functional is of order $O(N_{cluster}\lambda^2N^2)$, where $N_{cluster}< N$ is 
the number of test points for cluster-centers. 
On the other hand, the computational complexity of \textbf{Algorithm \ref{algo:standardmusic}} therein is $O(\lambda^3N^3)$ or $O(N_{cluster}\lambda^2N^2)$. We next consider \textbf{Algorithm \ref{algo:measurementdecouplebymultipole}}. Let $K$ be the number of clusters and $s$ be the number of multipole basis for each cluster. Then the size of the multipole matrix in \textbf{Algorithm \ref{algo:measurementdecouplebymultipole}} is $sK\times N$. Therefore, the involved computational complexity is of order $O((sK)^2N)$. Note that $sK \ll N$.  Finally, we consider \textbf{Algorithm \ref{algo:adaptiveregionbasedmusic}}. For each local measurement, let $N_{sub}\ll N$ be the number of samples used for reconstruction. The computational complexity of SVD in the MUSIC algorithm therein is of order $O(N_{sub}^3)$. Let $N_{source}$ be the number of test points for the point sources in the cluster. The computational complexity of constructing the MUSIC imaging functional is of order $O(N_{source} N_{sub}^2)$. 
Aggregating all these estimates, the computational complexity of \textbf{Algorithm \ref{algo:adaptivemeasurementdecouplingsupportrecovery}} is of order
$$
O(\lambda^3N^3+N_{cluster}\lambda^2N^2+(sK)^2N+KN_{sub}^3+KN_{source} N_{sub}^2) = O(\lambda^3N^3+N_{cluster}\lambda^2N^2).
$$

For comparison, we estimate the computational complexity of the standard MUSIC algorithm (\textbf{Algorithm \ref{algo:standardmusic}}). The computational complexity of SVD therein is of order $O(N^3)$ and of constructing MUSIC imaging functional is of order $O(N_{music}N^2)$, where $N_{music}$ is the number of test points. The computational complexity of standard MUSIC algorithm is of order $O(N^3+N_{music}N^2)$. Note that $N_{cluster} < N_{music}$ in practice. Therefore the computational complexity of \textbf{Algorithm \ref{algo:adaptivemeasurementdecouplingsupportrecovery}} is $\max(\lambda^3, \frac{N_{cluster}}{N_{music}}\lambda^2)$ of that of the standard MUSIC algorithm.

Finally, we conduct numerical experiments for the decoupling based super-resolution algorithm developed above.  We demonstrate that its super-resolving ability is comparable to that of the standard MUSIC algorithm but the time cost is substantially lower. 
We run $1000$ experiments with various clusters and reconstruct the source locations by D-MUSIC (\textbf{Algorithm \ref{algo:adaptivemeasurementdecouplingsupportrecovery}}) and the standard MUSIC algorithm (\textbf{Algorithm \ref{algo:standardmusic}}) respectively. We set $\Omega=1$, $\sigma = 10^{-3}$ and the sample number $N=1000$. We consider $(K, L, D, \Omega)$ regions with $L\geq 12\pi, D\approx \pi$ and signed measure $\mu =\sum_{j=1}^{K}\sum_{q=1}^{n_j}a_{q,j}\delta_{y_{q,j}}$ with random cluster number $K$. We only consider the case where each cluster has no more than three point sources and their separation distance are around $1$, which is smaller than the Rayleigh length $\pi$. These requirements make sure that the standard MUSIC algorithm can resolve all the point sources.  We choose shrinkage factor  $\lambda = \frac{1}{2}$ in \textbf{Algorithm \ref{algo:adaptivemeasurementdecouplingsupportrecovery}}. The results are shown in Figure \ref{fig:compareourtomptogrth}. Figure \ref{fig:compareourtomptogrth} (a) shows that the reconstruction error of the local measurement is of the noise level $10^{-3}$, which demonstrates the stability of the measurement decoupling algorithm. Figure \ref{fig:compareourtomptogrth} (b) plots the location recovery of \textbf{Algorithm \ref{algo:adaptivemeasurementdecouplingsupportrecovery}} and the standard MUSIC algorithm. It is shown that both algorithms can super-resolving all the point sources and the performances are comparable. Figure \ref{fig:compareourtomptogrth} (c) shows that the new algorithm is ten times faster than the standard MUSIC algorithm.


\begin{figure}[!h]
	\centering
	\begin{subfigure}[b]{0.28\textwidth}
		\centering
		\includegraphics[width=\textwidth]{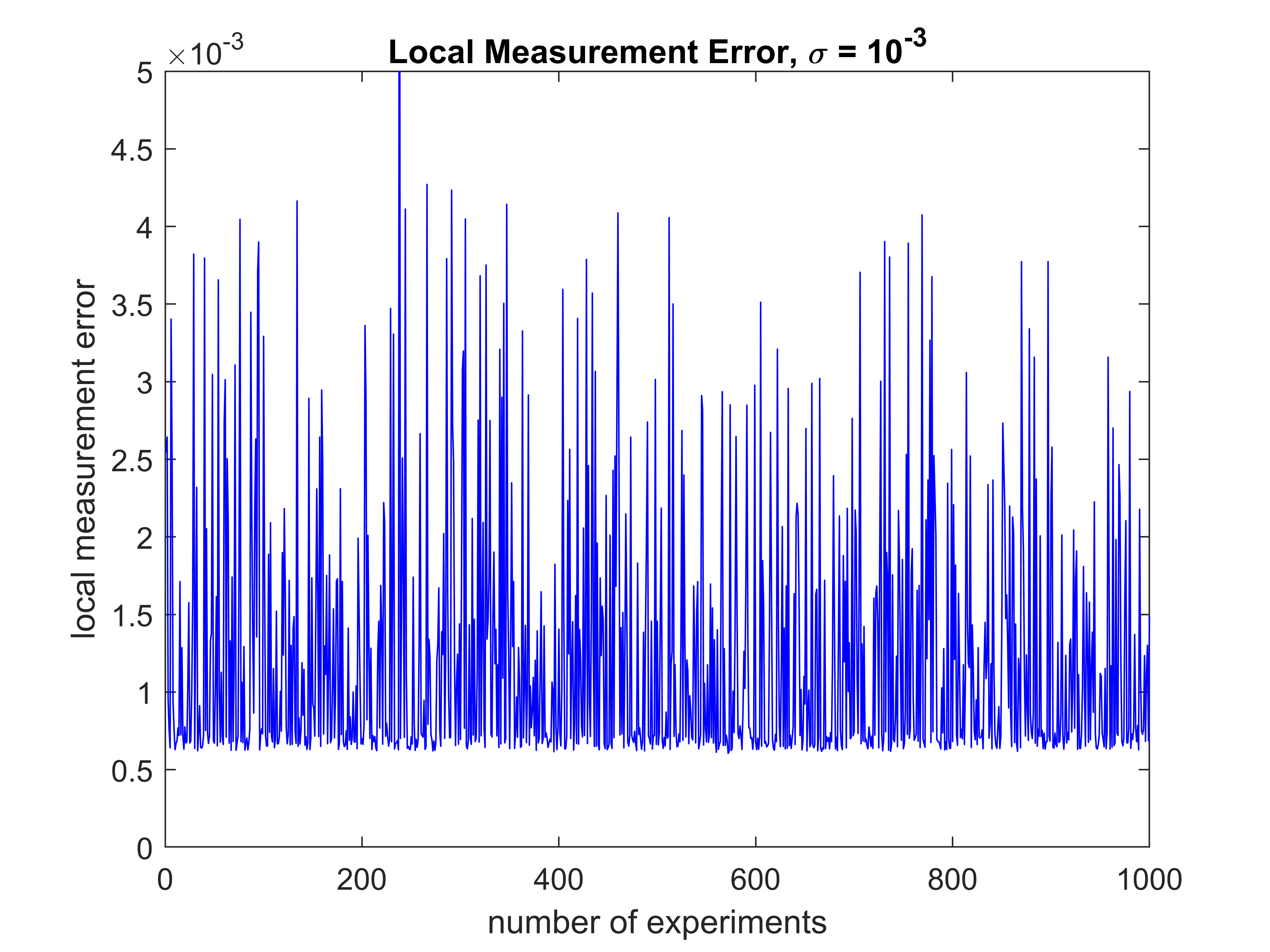}
		\caption{measurement recovery}
	\end{subfigure}
	\begin{subfigure}[b]{0.28\textwidth}
		\centering
		\includegraphics[width=\textwidth]{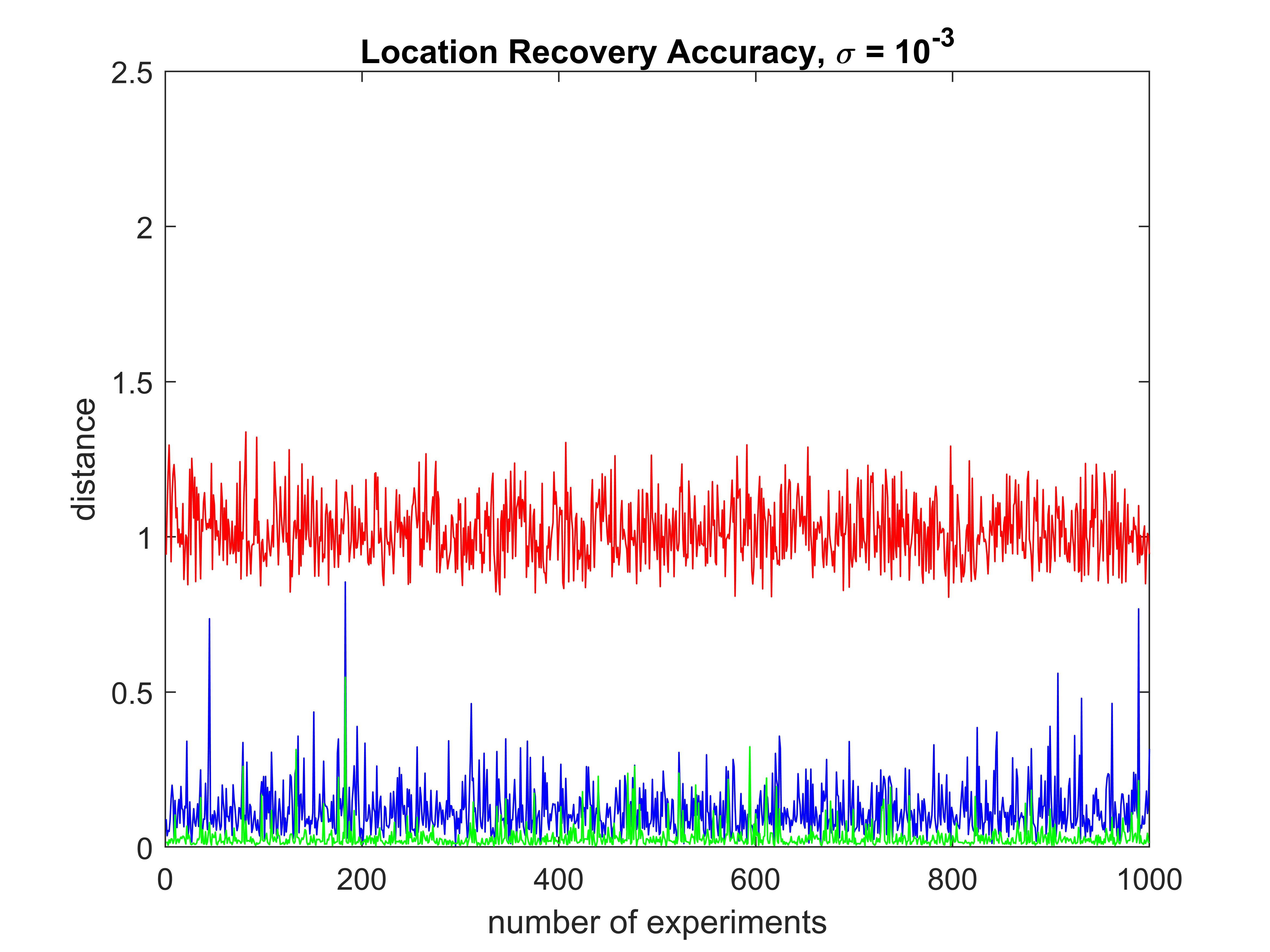}
		\caption{recovery accuracy}
	\end{subfigure}
	\begin{subfigure}[b]{0.28\textwidth}
		\centering
		\includegraphics[width=\textwidth]{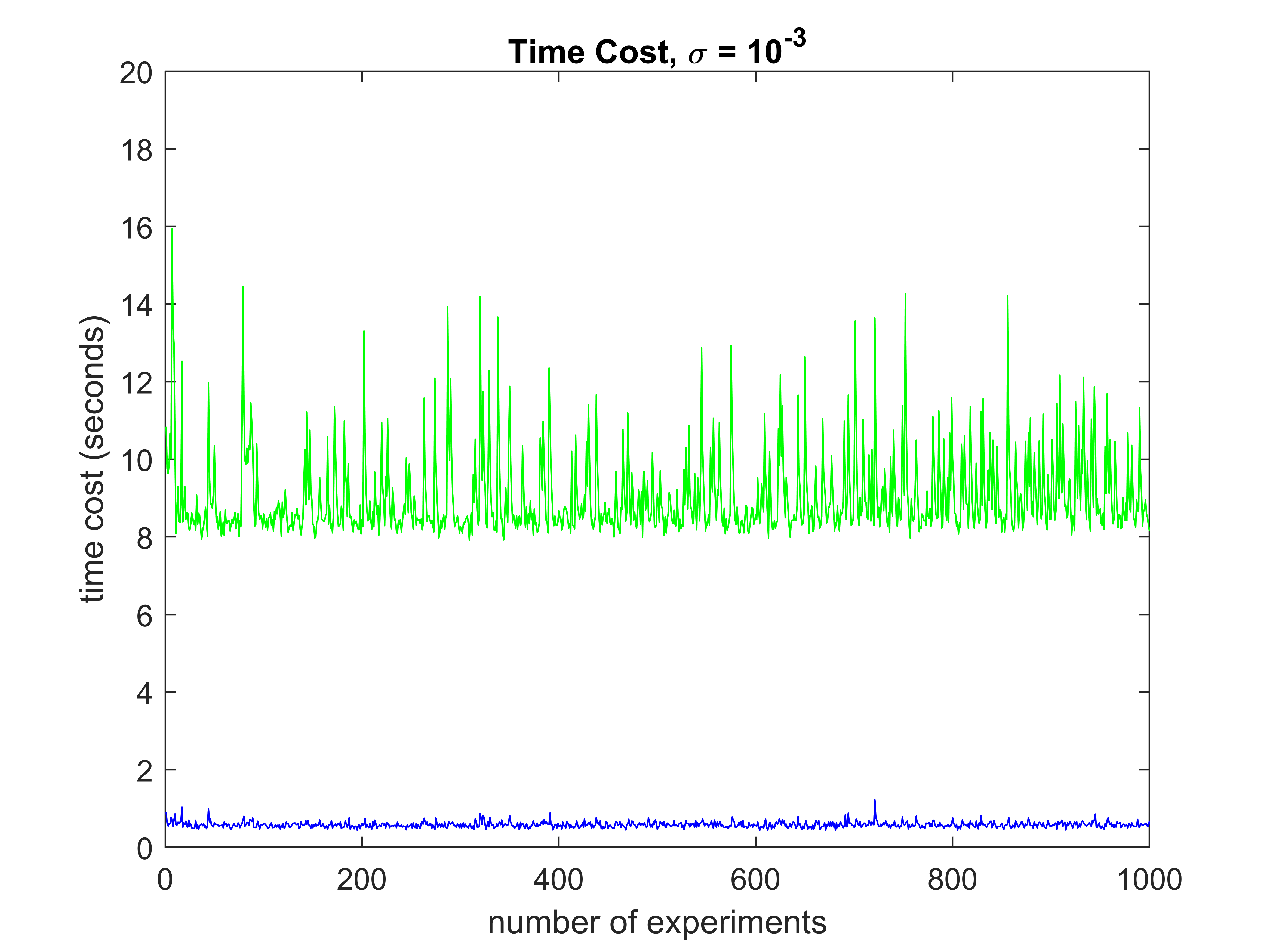}
		\caption{time cost}
	\end{subfigure}
	\caption{Plots of the accuracy of the measurement decoupling, the location recovery, and time cost by \textbf{Algorithm \ref{algo:adaptivemeasurementdecouplingsupportrecovery}} and \textbf{Algorithm \ref{algo:standardmusic}}. Figure (a) plots the accuracy of the recovered local measurements by \textbf{Algorithm \ref{algo:adaptivemeasurementdecouplingsupportrecovery}}. The blue line shows the maximum recovery error of local measurements in the experiments. The horizontal coordinate is the number of the experiments. Figure (b) plots the location recovery of the two algorithms. The red line  (around $1$) is the minimal separation distance of underlying point sources in each cluster and the blue line is the maximum deviation of the recovered locations by our algorithm to the ground truth. The green line is the one of the standard MUSIC algorithm. Figure (c) plots the time cost of the two algorithms. The green line is the time cost of the standard MUSIC algorithm and the blue line is the one of our algorithm.}
	\label{fig:compareourtomptogrth}
\end{figure}

\section{Conclusions and future works}
In this paper, we proposed an efficient algorithm, termed D-MUSIC, for super-resolving point sources with multi-cluster structure based on a measurement decoupling strategy. We demonstrated that the computational complexity of D-MUSIC is much lower than that of the standard MUSIC. There remains several interesting issues for future work. The first is on the estimate of the shrinkage factor $\lambda$ in the cluster structure detection algorithm. The choice of $\lambda$ plays a important role on the success of the algorithm. The second is on the optimal design of modulation function which can retain frequency components near the cut-off frequency while requiring less stringent condition on the separation distance of clusters for stable 
decoupling of local measurements. The last one is to extend the algorithm to the more general case when the cluster sizes may have a variety of scales. 

\section{Appendix}
\subsection{Proofs of some technical lemma}
We denote $||f||_{L_{\infty}([-1,1])}= \max_{x\in[-1,1]}|f(x)|$ and $||f||_{L_{1}([-1,1])}= \int_{a}^{b}|f|dx$ for continuous $f$.  

\begin{lem}(Markov brothers' inequality) \label{lem:polynomialineq1}
	Let $P_n(x)$ be a polynomial of degree at most $n$, we have
	\begin{align*}
		||P_{n}^{(k)}(x)||_{L_{\infty}([-1,1])}\leq \frac{n^2(n^2-1)\cdots(n^2-(k-1)^2)}{1\cdot 3\cdot 5\cdots (2k-1) }||P_{n}(x)||_{L_{\infty}([-1,1])}.
	\end{align*}
\end{lem}

\begin{lem}\label{lem:polynomialineq2}
	Let $P_n(x)$ be a polynomial of degree at most $n$, we have
	\[
	||P_{n}(x)||_{L_{\infty}([-1,1])}\leq (n+1)^2||P_{n}(x)||_{L_{1}([-1,1])}.
	\]
\end{lem}
\begin{proof}
	Given a polynomial $p$ of degree at most $n$, consider $q(x)=\int_{-1}^xp(t)dt$. We have $ |q(x)|\leq \int_{-1}^x|p(t)|dt\leq ||p||_{L_{1}([-1,1])}$ and therefore $||q||_{L_{\infty}([-1,1])}\leq||p||_{L_{1}([-1,1])}$. By Lemma \ref{lem:polynomialineq1}, one has $||q'||_{L_\infty([-1,1])}\leq (n+1)^2||q||_{L_\infty([-1,1])}$ since $q$ is of degree at most $n+1$. By $q' = p$, we have 
	\[
	||p||_{L_{\infty}([-1,1])}\leq (n+1)^2||p||_{L_{1}([-1,1])}.
	\]
\end{proof}

\begin{lem}\label{lem:oscillatoryintegral}
	Let $\psi$ be a polynomial of degree at most $n$, then for $\lambda \geq 2n^2$ we have
	\begin{align*}
		\Big|\int_{-1}^{1}e^{i\lambda x} \psi(x)dx\Big| < \frac{3.2||\psi||_{L_{\infty}([-1,1])}}{\lambda}.
	\end{align*}
	Suppose further that $\psi(\pm 1)=0, \psi'(\pm 1)=0$, we have 
	\begin{align*}
		\Big|\int_{-1}^{1}e^{i\lambda x} \psi(x)dx\Big| < \frac{0.8n^4||\psi||_{L_{\infty}([-1,1])}}{\lambda^3}.
	\end{align*}
\end{lem}
\begin{proof} By integration by parts, we have
	\begin{align*}
		\int_{-1}^{1}e^{i \lambda x} \psi(x)dx = \frac{e^{i\lambda x}\psi(x)}{i\lambda}\Big|_{-1}^1 -  \frac{\int_{-1}^{1}e^{i\lambda x}\psi'(x)dx}{i\lambda}= \frac{1}{i\lambda}\sum_{k=0}^n\frac{e^{i\lambda x}\psi^{(k)}(x)}{(-i\lambda)^{k}}\Big|_{-1}^1.
	\end{align*}
	Taking absolute value to both sides, we have 
	\begin{align*}
		&\Big|\int_{-1}^{1}e^{i \lambda x} \psi(x)dx\Big| \leq \frac{1}{\lambda} \Big(2\sum_{k=0}^n \frac{||\psi^{(k)}||_{L_{\infty}([-1,1])}}{\lambda^{k}}\Big)\\
		\leq& \frac{2||\psi||_{L_{\infty}([-1,1])}}{\lambda} \Big(1+\frac{n^2}{\lambda}+\frac{n^2(n^2-1)}{3\lambda^2}+\cdots +\frac{n^2(n^2-1)\cdots(n^2-(n-1)^2)}{(2n-1)!!\lambda^{n}}\Big) \quad \Big(\text{by Lemma \ref{lem:polynomialineq1}}\Big) \\
		<& \frac{3.2||\psi||_{L_{\infty}([-1,1])}}{\lambda}. \quad \Big(\text{by $\lambda \geq 2n^2$}\Big)
	\end{align*}
	Moreover, when $\psi(\pm 1)=0, \psi'(\pm 1)=0$, 
	\begin{align*}
		\int_{-1}^{1}e^{i \lambda x} \psi(x)dx =  \frac{1}{i\lambda}\sum_{k=2}^n\frac{e^{i\lambda x}\psi^{(k)}(x)}{(-i\lambda)^{k}}\Big|_{-1}^1.
	\end{align*}
	Taking absolute value to both sides, we have 
	\begin{align*}
		&\Big|\int_{-1}^{1}e^{i \lambda x} \psi(x)dx\Big| \leq \frac{1}{\lambda} \Big(2\sum_{k=2}^n \frac{||\psi^{(k)}||_{L_{\infty}([-1,1])}}{\lambda^{k}}\Big)\\
		\leq& \frac{2||\psi||_{L_{\infty}([-1,1])}}{\lambda} \Big( \frac{n^2(n^2-1)}{3\lambda^2}+\cdots +\frac{n^2(n^2-1)\cdots(n^2-(n-1)^2)}{(2n-1)!!\lambda^{n}}\Big) \quad \Big(\text{by Lemma \ref{lem:polynomialineq1}}\Big) \\
		<& \frac{0.8n^4||\psi||_{L_{\infty}([-1,1])}}{\lambda^3}, \quad \Big(\text{by $\lambda \geq 2n^2$}\Big)
	\end{align*}
	whence the lemma follows.
\end{proof}

\subsection{A peak selection algorithm}
\begin{algorithm}[H]
	\caption{\textbf{Peak selection algorithm}}
	\textbf{Input:} Image $IMG = (f(\omega_1), \cdots, f(\omega_M))$\;
	\textbf{Input:} Peak compare range $PCR$, differential compare range $DCR$,  differential compare threshold $DCT$\;
	1: Initialize the Local maximum points $LMP = [\ ]$, peak points $PP = [\ ]$\;
	2: Differentiate the image $IMG$ to get the $DIMG = (f'(\omega_1), \cdots, f'(\omega_M))$\;
	3: \For{$j=1:M$}{
		\If{$f(\omega_j) = \max(f(\omega_{j-PCR}), f(\omega_{j-PCR+1}), \cdots, f(\omega_{j+PCR}))$}{
			$LMP$ appends $\omega_j$\;
	}}
	4: \For{$\omega_j$ in $LMP$}{
		\If{
			$\max(|f'(\omega_{j-DCR})|, |f'(\omega_{j-DCR+1})|, \cdots, |f'(\omega_{j+DCR})|)\geq DCT$
		}{$PP$ appends $\omega_j$\;}
		\textbf{Return:} $PP$.
	}
	\label{algo:peakselection}
\end{algorithm}

\bibliographystyle{plain}
\bibliography{references} 

\begin{thebibliography}{10}

\bibitem{batenkov2021spectral}
Dmitry Batenkov, Benedikt Diederichs, Gil Goldman, and Yosef Yomdin.
\newblock The spectral properties of vandermonde matrices with clustered nodes.
\newblock {\em Linear Algebra and its Applications}, 609:37--72, 2021.

\bibitem{batenkov2019super}
Dmitry Batenkov, Gil Goldman, and Yosef Yomdin.
\newblock {Super-resolution of near-colliding point sources}.
\newblock {\em Information and Inference: A Journal of the IMA}, 05 2020.
\newblock iaaa005.

\bibitem{bernstein2019deconvolution}
Brett Bernstein and Carlos Fernandez-Granda.
\newblock Deconvolution of point sources: a sampling theorem and robustness
  guarantees.
\newblock {\em Communications on Pure and Applied Mathematics},
  72(6):1152--1230, 2019.

\bibitem{cai2019fast}
Jian-Feng Cai, Tianming Wang, and Ke~Wei.
\newblock Fast and provable algorithms for spectrally sparse signal
  reconstruction via low-rank hankel matrix completion.
\newblock {\em Applied and Computational Harmonic Analysis}, 46(1):94--121,
  2019.

\bibitem{candes2013super}
Emmanuel~J. Cand{\`e}s and Carlos Fernandez-Granda.
\newblock Super-resolution from noisy data.
\newblock {\em Journal of Fourier Analysis and Applications}, 19(6):1229--1254,
  2013.

\bibitem{candes2014towards}
Emmanuel~J. Cand{\`e}s and Carlos Fernandez-Granda.
\newblock Towards a mathematical theory of super-resolution.
\newblock {\em Communications on Pure and Applied Mathematics}, 67(6):906--956,
  2014.

\bibitem{chi2020harnessing}
Yuejie Chi and Maxime~Ferreira Da~Costa.
\newblock Harnessing sparsity over the continuum: Atomic norm minimization for
  superresolution.
\newblock {\em IEEE Signal Processing Magazine}, 37(2):39--57, 2020.

\bibitem{denoyelle2017support}
Quentin Denoyelle, Vincent Duval, and Gabriel Peyr{\'e}.
\newblock Support recovery for sparse super-resolution of positive measures.
\newblock {\em Journal of Fourier Analysis and Applications}, 23(5):1153--1194,
  2017.

\bibitem{duval2015exact}
Vincent Duval and Gabriel Peyr{\'e}.
\newblock Exact support recovery for sparse spikes deconvolution.
\newblock {\em Foundations of Computational Mathematics}, 15(5):1315--1355,
  2015.

\bibitem{eldar2010block}
Yonina~C Eldar, Patrick Kuppinger, and Helmut Bolcskei.
\newblock Block-sparse signals: Uncertainty relations and efficient recovery.
\newblock {\em IEEE Transactions on Signal Processing}, 58(6):3042--3054, 2010.

\bibitem{eldar2009block}
Yonina~C Eldar and Moshe Mishali.
\newblock Block sparsity and sampling over a union of subspaces.
\newblock In {\em 2009 16th International Conference on Digital Signal
  Processing}, pages 1--8. IEEE, 2009.

\bibitem{eldar2009robust}
Yonina~C Eldar and Moshe Mishali.
\newblock Robust recovery of signals from a structured union of subspaces.
\newblock {\em IEEE Transactions on Information Theory}, 55(11):5302--5316,
  2009.

\bibitem{3cd521e14db44a81a63d9d80c4aaa1ca}
Carlos Fernandez-Granda.
\newblock Support detection in super-resolution.
\newblock In {\em Proceedings of the 10th International Conference on Sampling
  Theory and Applications (SampTA 2013)}, pages 145--148, 2013.

\bibitem{han2013improved}
Keyong Han and Arye Nehorai.
\newblock Improved source number detection and direction estimation with nested
  arrays and ulas using jackknifing.
\newblock {\em IEEE Transactions on Signal Processing}, 61(23):6118--6128,
  2013.

\bibitem{he2010detecting}
Zhaoshui He, Andrzej Cichocki, Shengli Xie, and Kyuwan Choi.
\newblock Detecting the number of clusters in n-way probabilistic clustering.
\newblock {\em IEEE Transactions on Pattern Analysis and Machine Intelligence},
  32(11):2006--2021, 2010.

\bibitem{hua1990matrix}
Yingbo Hua and Tapan~K. Sarkar.
\newblock Matrix pencil method for estimating parameters of exponentially
  damped/undamped sinusoids in noise.
\newblock {\em IEEE Transactions on Acoustics, Speech, and Signal Processing},
  38(5):814--824, 1990.

\bibitem{li2021stable}
Weilin Li and Wenjing Liao.
\newblock Stable super-resolution limit and smallest singular value of
  restricted fourier matrices.
\newblock {\em Applied and Computational Harmonic Analysis}, 51:118--156, 2021.

\bibitem{li2019super}
Weilin Li, Wenjing Liao, and Albert Fannjiang.
\newblock Super-resolution limit of the esprit algorithm.
\newblock {\em IEEE Transactions on Information Theory}, 66(7):4593--4608.

\bibitem{liao2016music}
Wenjing Liao and Albert~C. Fannjiang.
\newblock Music for single-snapshot spectral estimation: Stability and
  super-resolution.
\newblock {\em Applied and Computational Harmonic Analysis}, 40(1):33--67,
  2016.

\bibitem{liu2021mathematicalhighd}
Ping Liu and Hai Zhang.
\newblock A mathematical theory of computational resolution limit in
  multi-dimensional spaces.
\newblock {\em Inverse Problems}, 37(10):104001, 2021.

\bibitem{liu2021theorylse}
Ping Liu and Hai Zhang.
\newblock A theory of computational resolution limit for line spectral
  estimation.
\newblock {\em IEEE Transactions on Information Theory}, 67(7):4812--4827,
  2021.

\bibitem{liu2021mathematicaloned}
Ping Liu and Hai Zhang.
\newblock A mathematical theory of the computational resolution limit in one
  dimension.
\newblock {\em Applied and Computational Harmonic Analysis}, 56:402--446, 2022.

\bibitem{morgenshtern2020super}
Veniamin~I Morgenshtern.
\newblock Super-resolution of positive sources on an arbitrarily fine grid.
\newblock {\em arXiv preprint arXiv:2005.06756}, 2020.

\bibitem{morgenshtern2016super}
Veniamin~I. Morgenshtern and Emmanuel~J. Candes.
\newblock Super-resolution of positive sources: The discrete setup.
\newblock {\em SIAM Journal on Imaging Sciences}, 9(1):412--444, 2016.

\bibitem{roy1989esprit}
Richard Roy and Thomas Kailath.
\newblock Esprit-estimation of signal parameters via rotational invariance
  techniques.
\newblock {\em IEEE Transactions on acoustics, speech, and signal processing},
  37(7):984--995, 1989.

\bibitem{schmidt1986multiple}
Ralph Schmidt.
\newblock Multiple emitter location and signal parameter estimation.
\newblock {\em IEEE transactions on antennas and propagation}, 34(3):276--280,
  1986.

\bibitem{stojnic2009reconstruction}
Mihailo Stojnic, Farzad Parvaresh, and Babak Hassibi.
\newblock On the reconstruction of block-sparse signals with an optimal number
  of measurements.
\newblock {\em IEEE Transactions on Signal Processing}, 57(8):3075--3085, 2009.

\bibitem{tang2015resolution}
Gongguo Tang.
\newblock Resolution limits for atomic decompositions via markov-bernstein type
  inequalities.
\newblock In {\em 2015 International Conference on Sampling Theory and
  Applications (SampTA)}, pages 548--552. IEEE, 2015.

\bibitem{tang2014near}
Gongguo Tang, Badri~Narayan Bhaskar, and Benjamin Recht.
\newblock Near minimax line spectral estimation.
\newblock {\em IEEE Transactions on Information Theory}, 61(1):499--512, 2014.

\bibitem{wang2008tree}
Yung-Yi Wang, Liang-Cheng Lee, Shih-Jen Yang, and Jiunn-Tsair Chen.
\newblock A tree structure one-dimensional based algorithm for estimating the
  two-dimensional direction of arrivals and its performance analysis.
\newblock {\em IEEE transactions on antennas and propagation}, 56(1):178--188,
  2008.

\bibitem{wax1985detection}
Mati Wax and Thomas Kailath.
\newblock Detection of signals by information theoretic criteria.
\newblock {\em IEEE Transactions on acoustics, speech, and signal processing},
  33(2):387--392, 1985.

\end{thebibliography}
\end{document}